%% file: main.tex
\documentclass[a4paper]{article}
\usepackage[utf8x]{inputenc}

\include{urlsetup}
\include{tabsetup}
\include{mathsetup}

\include{refsetup}
\include{tikzsetup}
\include{macrosetup}

\usepackage{comment}
\usepackage{floatflt}
\usepackage{enumitem}

\title{Counting and Enumerating\\Crossing-free Geometric Graphs}
\author{Manuel Wettstein\footnote{Department of Computer Science, ETH Z\"urich, Switzerland. E-mail: \url{mw@inf.ethz.ch}}}
\date{\today}

\begin{document}

\maketitle

\begin{abstract}
  \input{abstract}
\end{abstract}

\input{introduction}
\input{basics}

\input{framework}

\input{applications}
\input{construction}
\input{enumeration}
\input{spanning}
\input{bounds}
\input{acknowledgement}

\bibliographystyle{acm}
\bibliography{refs}



\end{document}

%% file: urlsetup.tex
\usepackage{url}

\newcommand{\emailaddress}[1]{\href{mailto:#1}{\nolinkurl{#1}}}

%% file: tabsetup.tex
\usepackage{booktabs}

%% file: mathsetup.tex
\usepackage{amsmath}
\usepackage{amssymb}
\usepackage{amsthm}
\usepackage{mathtools}
\usepackage{thm-restate}

\newtheorem{definition}{Definition}
\newtheorem{observation}{Observation}

\newtheorem{lemma}{Lemma}
\newtheorem{theorem}{Theorem}
\newtheorem{corollary}{Corollary}

%% file: refsetup.tex
\usepackage[hidelinks]{hyperref}
\usepackage[capitalise]{cleveref}

\crefname{claim}{Claim}{Claims}
\crefname{observation}{Observation}{Observations}
\crefname{ghost}{condition}{conditions}

\crefname{figure}{Figure}{Figures}

%% file: tikzsetup.tex
\usepackage{tikz}
\usetikzlibrary{calc}
\usetikzlibrary{shapes}
\usetikzlibrary{patterns}
\usetikzlibrary{arrows}
\usetikzlibrary{positioning}
\usetikzlibrary{decorations.markings}
\usetikzlibrary{backgrounds}

\tikzstyle{thinpoint} = [
  circle,
  inner sep=0.9pt,
  draw=black,
  fill=black
]
\tikzstyle{point} = [
  circle,
  inner sep=1.2pt,
  draw=black,
  fill=black
]
\tikzstyle{fatpoint} = [
  circle,
  inner sep=1.5pt,
  draw=black,
  fill=white
]
\tikzstyle{cross} = [
  path picture = {
    \draw[black]
      (path picture bounding box.15) --
      (path picture bounding box.255);
    \draw[black]
      (path picture bounding box.45) --
      (path picture bounding box.225);
    \draw[black]
      (path picture bounding box.75) --
      (path picture bounding box.195);
  }
]
\tikzstyle{free} = [
  fatpoint,
  fill=white
]
\tikzstyle{alive} = [
  fatpoint,
  fill=white,
  cross
]
\tikzstyle{dead} = [
  fatpoint,
  fill=black
]

\tikzstyle{collabel} = [
  path picture = {
    \node {\tiny #1};
  }
]

\tikzstyle{colone} = [
  fatpoint,
  fill=white,
  collabel=1
]
\tikzstyle{coltwo} = [
  fatpoint,
  fill=white,
  collabel=2
]
\tikzstyle{colthree} = [
  fatpoint,
  fill=white,
  collabel=3
]
\tikzstyle{colfour} = [
  fatpoint,
  fill=white,
  collabel=4
]
\tikzstyle{colfive} = [
  fatpoint,
  fill=white,
  collabel=5
]
\tikzstyle{colsix} = [
  fatpoint,
  fill=white,
  collabel=6
]
\tikzstyle{colseven} = [
  fatpoint,
  fill=white,
  collabel=7
]
\tikzstyle{coleight} = [
  fatpoint,
  fill=white,
  collabel=8
]

\tikzstyle{veryfatpoint} = [
  circle,
  inner sep=2.0pt,
  fill=white,
  draw=black
]

\tikzstyle{drawL} = [
  path picture = {
    \draw[fill=black,scale=1.2]
      (60:0.04) -- 
      (180:0.04) -- 
      (300:0.04) -- 
      cycle;
  }
]
\tikzstyle{drawR} = [
  path picture = {
    \draw[fill=black]
      (0:0.04) -- 
      (120:0.04) -- 
      (240:0.04) -- 
      cycle;
  }
]

\tikzstyle{drlft} = [
  veryfatpoint,
  drawL
]
\tikzstyle{drrgt} = [
  veryfatpoint,
  drawR
]
\tikzstyle{marked} = [
  circle,
  inner sep=0.5pt,
  draw=black,
  fill=black
]
\tikzstyle{unmarked} = [
  inner sep=1.5pt
]

\newcommand\inlinepoint[2]{
  \begin{tikzpicture}
    \node[#1,scale=1.2] (point) at (0,0) {};
    \node[#2,scale=1.2] (marking) at (0,0) {};
  \end{tikzpicture}
}

\newcommand{\ptfree}{\inlinepoint{free}{unmarked}}
\newcommand{\ptalive}{\inlinepoint{alive}{unmarked}}
\newcommand{\ptaliveM}{\inlinepoint{alive}{marked}}
\newcommand{\ptdead}{\inlinepoint{dead}{unmarked}}

\newcommand{\ptone}{\inlinepoint{colone}{unmarked}}
\newcommand{\pttwo}{\inlinepoint{coltwo}{unmarked}}
\newcommand{\ptthree}{\inlinepoint{colthree}{unmarked}}
\newcommand{\ptfour}{\inlinepoint{colfour}{unmarked}}
\newcommand{\ptfive}{\inlinepoint{colfive}{unmarked}}
\newcommand{\ptsix}{\inlinepoint{colsix}{unmarked}}
\newcommand{\ptseven}{\inlinepoint{colseven}{unmarked}}
\newcommand{\pteight}{\inlinepoint{coleight}{unmarked}}

\newcommand{\ptdrlft}{\inlinepoint{drlft}{unmarked}}
\newcommand{\ptdrrgt}{\inlinepoint{drrgt}{unmarked}}

\tikzstyle{arrowstyle} = [scale=1]
\tikzstyle{oriented} = [
  postaction = {decorate, decoration = {
    markings,
    mark = at position .65 with {\arrow[arrowstyle]{stealth};}}}
]
\tikzstyle{dloriented} = [
  postaction = {decorate, decoration = {
    markings,
    mark = at position .5 with {
      \draw[solid,->,thin] (-0.03,-0.07) -- (-0.03,+0.1);
  }}}
]
\tikzstyle{droriented} = [
  postaction = {decorate, decoration = {
    markings,
    mark = at position .5 with {
      \draw[solid,->,thin] (-0.03,+0.07) -- (-0.03,-0.1);
  }}}
]

%% file: macrosetup.tex
\newcommand{\pts}{P}                              
\newcommand{\pt}{p}                               
\newcommand{\ptsnum}{n}                           
\newcommand{\ptspreceq}{\preceq}                  

\newcommand{\pgs}{\mathcal{G}(\pts)}              
\newcommand{\trs}{\mathcal{T}(\pts)}              
\newcommand{\css}{\mathcal{CS}(\pts)}             
\newcommand{\pms}{\mathcal{M}(\pts)}              
\newcommand{\cps}{\mathcal{CP}(\pts)}             
\newcommand{\sts}{\mathcal{ST}(\pts)}             
\newcommand{\scs}{\mathcal{SC}(\pts)}             

\newcommand{\pmss}{\mathcal{EM}(\pts)}             

\newcommand{\pgsnum}{\operatorname{pg}(\pts)}     
\newcommand{\trsnum}{\operatorname{tr}(\pts)}     
\newcommand{\cssnum}{\operatorname{cs}(\pts)}     
\newcommand{\pmsnum}{\operatorname{pm}(\pts)}     
\newcommand{\cpsnum}{\operatorname{cp}(\pts)}     
\newcommand{\stsnum}{\operatorname{st}(\pts)}     
\newcommand{\scsnum}{\operatorname{sc}(\pts)}     

\newcommand{\segs}{{\mathcal{S}_\pts}}            
\newcommand{\cparts}{\mathcal{CP}_{\pts}}         
\newcommand{\cpartstp}{\mathcal{CF}_{\pts}}       
\newcommand{\trgs}{{\mathcal{T}_\pts}}            

\newcommand{\segsst}{\mathcal{S}_\pts^\tagst}     
\newcommand{\segssc}{\mathcal{S}_\pts^\tagsc}     

\newcommand{\frags}{\mathcal{G}^*_\pts}           

\newcommand{\unit}{u}                             
\newcommand{\units}{\mathcal{U}}                  
\newcommand{\cmb}{C}                              
\newcommand{\umaxin}[1]{\operatorname{rex}(#1)}   

\newcommand{\upts}[1]{\operatorname{pts}(#1)}
\newcommand{\ulow}[1]{\operatorname{low}(#1)}
\newcommand{\uupp}[1]{\operatorname{upp}(#1)}
\newcommand{\ulft}[1]{\operatorname{lft}(#1)}
\newcommand{\urgt}[1]{\operatorname{rgt}(#1)}

\newcommand{\uslow}[1]{\operatorname{\overline{low}}(#1)}

\newcommand{\udep}{\sqsubset}
\newcommand{\undep}{\not\sqsubset}
\newcommand{\utorgt}{\ptspreceq}

\newcommand{\cpts}[1]{\operatorname{pts}(#1)}
\newcommand{\clow}[1]{\operatorname{low}(#1)}

\newcommand{\cchn}[1]{\operatorname{chn}(#1)}     

\newcommand{\parteq}{\sim}                      
\newcommand{\partclass}[1]{[{#1}]}                
\newcommand{\quotient}[2]{({#1} \big/ {#2})}      
\newcommand{\target}{\mathfrak{T}}                

\newcommand\targetcf[1]{\target^{\tagcf}_{\units}}
\newcommand\targetpt[1]{\target^{\tagpt}_{\units}}
\newcommand\targetsd[1]{\target^{\tagsd}_{\units}}

\newcommand{\cgraph}{\Gamma}                      
\newcommand{\cgsource}{\bot}                      
\newcommand{\cgsink}{\top}                        

\newcommand{\agraphcf}[1]{\cgraph^\tagcf_{#1}}    
\newcommand{\agraphpt}[1]{\cgraph^\tagpt_{#1}}    
\newcommand{\agraphsd}[1]{\cgraph^\tagsd_{#1}}    
\newcommand{\agraphst}[1]{\cgraph^\tagst_{#1}}    
\newcommand{\agraphsc}[1]{\cgraph^\tagsc_{#1}}    

\newcommand{\tagcf}{{\mathsf{cf}}}                
\newcommand{\tagpt}{{\mathsf{pt}}}                
\newcommand{\taglpt}{{\mathsf{\check{\raisebox{0pt}[3.0pt]{$\mathsf{\scriptstyle pt}$}}}}}
\newcommand{\tagsd}{{\mathsf{sd}}}                
\newcommand{\taglsd}{{\mathsf{\check{\raisebox{0pt}[3.0pt]{$\mathsf{\scriptstyle sd}$}}}}}
\newcommand{\tagst}{{\mathsf{st}}}                
\newcommand{\taglst}{{\mathsf{\check{\raisebox{0pt}[3.0pt]{$\mathsf{\scriptstyle st}$}}}}}
\newcommand{\tagsc}{{\mathsf{sc}}}                
\newcommand{\taglsc}{{\mathsf{\check{\raisebox{0pt}[3.0pt]{$\mathsf{\scriptstyle sc}$}}}}}

\newcommand{\cmbscf}[1]{\cmbs^\tagcf(#1)}         
\newcommand{\cmbspt}[1]{\cmbs^\tagpt(#1)}         
\newcommand{\cmbslpt}[1]{\cmbs^\taglpt(#1)}       
\newcommand{\cmbssd}[1]{\cmbs^\tagsd(#1)}         
\newcommand{\cmbslsd}[1]{\cmbs^\taglsd(#1)}       
\newcommand{\cmbsst}{\cmbs^\tagst_\pts}           
\newcommand{\cmbslst}{\cmbs^\taglst_\pts}         
\newcommand{\cmbssc}{\cmbs^\tagsc_\pts}           
\newcommand{\cmbslsc}{\cmbs^\taglsc_\pts}         

\newcommand{\cmbs}{\mathfrak{C}}

\newcommand{\ctarrow}[3]{#1 \xrightarrow{\smash{#3}} #2}
\newcommand{\ctnarrow}[3]{#1 \not\xrightarrow{} #2}

\newcommand{\cgarrow}[3]{\partclass{#1} \xrightarrow{\smash{#3}} \partclass{#2}}

\newcommand{\cgpath}{R}
\newcommand{\cmbntr}[1]{\cmb_{#1}}

%% file: abstract.tex
We describe a framework for counting and enumerating various types of crossing-free geometric graphs on a planar point set.
The framework generalizes ideas of Alvarez and Seidel, who used them to count triangulations in time $O(2^nn^2)$ where $n$ is the number of points.
The main idea is to reduce the problem of counting geometric graphs to counting source-sink paths in a directed acyclic graph.

The following new results will emerge.
The number of all crossing-free geometric graphs can be computed in time $O(c^nn^4)$ for some $c < 2.83929$.
The number of crossing-free convex partitions can be computed in time $O(2^nn^4)$.
The number of crossing-free perfect matchings can be computed in time $O(2^nn^4)$.
The number of convex subdivisions can be computed in time $O(2^nn^4)$.
The number of crossing-free spanning trees can be computed in time $O(c^nn^4)$ for some $c < 7.04313$.
The number of crossing-free spanning cycles can be computed in time $O(c^nn^4)$ for some $c < 5.61804$.

With the same bounds on the running time we can construct data structures which allow fast enumeration of the respective classes.
For example, after $O(2^nn^4)$ time of preprocessing we can enumerate the set of all crossing-free perfect matchings using polynomial time per enumerated object.
For crossing-free perfect matchings and convex partitions we further obtain enumeration algorithms where the time delay for each (in particular, the first) output is bounded by a polynomial in $n$.

All described algorithms are comparatively simple, both in terms of their analysis and implementation.

%% file: introduction.tex
\section{Introduction}\label{sec:introduction}

Let $\pts$ be a set of $\ptsnum$ points in the plane.
We assume $\pts$ to be in \emph{general position}, which means that no three points in $\pts$ are collinear.
A \emph{geometric graph} on $\pts$ is a simple graph with vertex set $\pts$, combined with an embedding into the plane where edges are drawn as straight segments between the corresponding endpoints.
Two distinct edges are \emph{crossing} if their drawings intersect in their respective relative interiors, otherwise they are \emph{non-crossing}.
A geometric graph on $\pts$ is \emph{crossing-free} if its edges are pairwise non-crossing.

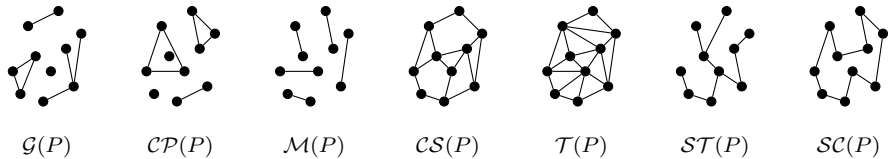
\begin{figure}
  \begin{center}
    \input{fig/graphexamples}
  \end{center}
  \caption{
    The seven defined classes of crossing-free geometric graphs.
  }
  \label{fig:graphexamples}
\end{figure}

We define the set $\pgs$ of all crossing-free geometric graphs on $\pts$, the set $\cps$ of \emph{crossing-free convex partitions}, the set $\pms$ of \emph{crossing-free perfect matchings}, the set $\css$ of \emph{convex subdivisions}, the set $\trs$ of \emph{triangulations}, the set $\sts$ of \emph{crossing-free spanning trees}, and the set $\scs$ of \emph{crossing-free spanning cycles}.
\Cref{fig:graphexamples} shows one representative of each defined class on a fixed point set.
Note that by a crossing-free convex partition we mean a partition of $\pts$ in such a way that the convex hulls of individual parts are pairwise disjoint.
Convex subdivisions are subdivisions of the convex hull of $\pts$ into convex faces with points from $\pts$ as vertices.
Note that every crossing-free convex partition and every convex subdivision is represented uniquely by a crossing-free geometric graph which contains all the edges on the boundaries of individual parts and faces.
Finally, a triangulation is an edge-maximal crossing-free geometric graph, which is a special case of a convex subdivision.
\begin{table*}[b]
  \begin{center}
    \input{tab/extremalbounds}
  \end{center}
  \caption{
    Extremal bounds, where cells display the respective exponential bases.
  }
  \label{tab:extremalbounds}
\end{table*}
\footnotetext{Apply the lower bound for perfect matchings on all subsets of $\pts$ and then use the binomial theorem.}

Following a common notation, we also define the numbers $\pgsnum \coloneqq |\pgs|$, $\cpsnum \coloneqq |\cps|$, $\pmsnum \coloneqq |\pms|$, $\cssnum \coloneqq |\css|$, $\trsnum \coloneqq |\trs|$, $\stsnum \coloneqq |\sts|$, and $\scsnum \coloneqq |\scs|$.
Some of these numbers have received considerable attention.
In particular, exponential upper and lower bounds have been established and gradually improved over the past decades.
Instead of listing all successive improvements, which are too numerous to count, we summarize the current state of affairs in \cref{tab:extremalbounds}.
As a reading example, the top-left entry says that $\pgsnum = O^*(187.53^ n)$ holds for all point sets $\pts$, see \cite{SS12} (the $*$ indicates that any subexponential factors are ignored).
For an always up-to-date list, which includes bounds for many other types of geometric graphs, we refer the interested reader to Sheffer's webpage \cite{ShWeb}.

The defined classes of geometric graphs have also been studied from an algorithmic point of view.
The problem of enumeration has been solved for $\pgs$, $\trs$ and $\sts$, see \cite{AF96,B02,AAHV07,KT08}.
By solved we mean that these sets can be enumerated in such a way that the time delay for each enumerated object is bounded by a polynomial in $\ptsnum$.

In terms of counting, especially triangulations have been studied extensively \cite{A99,ABRS15,KLS15}.
Furthermore, for counting many other types of geometric graphs there already exists a general framework \cite{ABCR12}, which is based on the onion layer structure of a point set.
However, for a long time no counting algorithm was always provably faster than enumerating the set whose size was to be determined.
For triangulations this changed with a remarkable paper by Alvarez and Seidel \cite{SA13}, who showed how to compute the number $\trsnum$ in time $O(2^nn^ 2)$.
This is always exponentially faster than enumeration because $\trsnum = \Omega^*(2.43^\ptsnum)$ holds for all sets $\pts$, see \cref{tab:extremalbounds}.
An unrelated approach led to a similar result for the class of all crossing-free geometric graphs.
Razen and Welzl \cite{RW11} showed how to compute $\pgsnum$ in time $O^\ast(\trsnum)$.
Since they proved that $\pgsnum = \Omega^*(2.82^n) \cdot \trsnum$ holds for all sets $\pts$, their algorithm also achieves an exponential speed-up over any procedure that counts by enumerating the whole set $\pgs$.
Furthermore, after the first publication of the present paper \cite{W14}, Marx and Miltzow discovered a new class of algorithms which make use of the so called cactus layer structure of a triangulation \cite{MM16}.
These algorithms run in time $n^{O(\sqrt n)}$ and are able to compute the number of triangulations and related structures exactly.

In the remainder of this paper we develop and make use of an abstract framework that generalizes ideas originally used by Alvarez and Seidel \cite{SA13} for counting triangulations.
Loosely speaking, the technique boils down to the following steps.
Fix any set of crossing-free geometric graphs whose elements can be decomposed into reasonably small or simple pieces.
For instance, every triangulation can be decomposed into a set of interior-disjoint triangles and, similarly, each crossing-free perfect matching can be decomposed into a set of non-intersecting segments.
The aim then is to construct a directed acyclic graph $\cgraph$ with the following properties.
Firstly, each edge in $\cgraph$ is labeled with one of the aforementioned pieces.
Secondly, there exist distinguished source and sink vertices in $\cgraph$.
Thirdly, there is a natural bijection between source-sink paths in $\cgraph$ and the fixed set of geometric graphs.
By this we mean that given any source-sink path in $\cgraph$, we can collect all the labels appearing on that path and combine them to obtain the corresponding geometric graph.
Clearly, given such a graph $\cgraph$ for one particular class of geometric graphs, the corresponding counting and enumeration problems can be reduced to counting and enumerating source-sink paths in $\cgraph$.

%% file: fig/graphexamples.tex
\begin{tikzpicture}
  \newcommand\examplePoints[1]{
    \node[point] (p0) at (0.0,0.6) {};
    \node[point] (p1) at (0.1,0.3) {};
    \node[point] (p2) at (0.2,1.2) {};
    \node[point] (p3) at (0.3,0.8) {};
    \node[point] (p4) at (0.4,0.2) {};
    \node[point] (p5) at (0.5,0.6) {};
    \node[point] (p6) at (0.6,1.4) {};
    \node[point] (p7) at (0.7,0.9) {};
    \node[point] (p8) at (0.8,0.4) {};
    \node[point] (p9) at (0.9,1.1) {};
    \node (label) at (0.45,-0.4) {\footnotesize #1};
  }
  \begin{scope}[xshift=0*160]
    \examplePoints{$\pgs$}
    \draw (p0) -- (p1);
    \draw (p1) -- (p3);
    \draw (p0) -- (p3);
    \draw (p2) -- (p6);
    \draw (p4) -- (p8);
    \draw (p7) -- (p8);
    \draw (p8) -- (p9);
  \end{scope}
  \begin{scope}[xshift=3*50]
    \examplePoints{$\css$}
    \draw (p0) -- (p1);
    \draw (p0) -- (p2);
    \draw (p1) -- (p4);
    \draw (p2) -- (p6);
    \draw (p4) -- (p8);
    \draw (p6) -- (p9);
    \draw (p8) -- (p9);
    \draw (p3) -- (p7);
    \draw (p0) -- (p3);
    \draw (p2) -- (p3);
    \draw (p7) -- (p8);
    \draw (p7) -- (p9);
    \draw (p3) -- (p5);
    \draw (p4) -- (p5);
    \draw (p7) -- (p5);
  \end{scope}
  \begin{scope}[xshift=4*50]
    \examplePoints{$\trs$}
    \draw (p0) -- (p1);
    \draw (p0) -- (p2);
    \draw (p1) -- (p4);
    \draw (p2) -- (p6);
    \draw (p4) -- (p8);
    \draw (p6) -- (p9);
    \draw (p8) -- (p9);
    \draw (p3) -- (p5);
    \draw (p3) -- (p7);
    \draw (p5) -- (p7);
    \draw (p2) -- (p9);
    \draw (p0) -- (p3);
    \draw (p0) -- (p5);
    \draw (p1) -- (p5);
    \draw (p2) -- (p3);
    \draw (p2) -- (p7);
    \draw (p4) -- (p5);
    \draw (p5) -- (p8);
    \draw (p7) -- (p8);
    \draw (p7) -- (p9);
  \end{scope}
  \begin{scope}[xshift=1*50]
    \examplePoints{$\cps$}
    \draw (p0) -- (p2);
    \draw (p0) -- (p5);
    \draw (p2) -- (p5);
    \draw (p4) -- (p8);
    \draw (p6) -- (p7);
    \draw (p6) -- (p9);
    \draw (p7) -- (p9);
  \end{scope}
  \begin{scope}[xshift=2*50]
    \examplePoints{$\pms$}
    \draw (p0) -- (p5);
    \draw (p1) -- (p4);
    \draw (p2) -- (p3);
    \draw (p6) -- (p7);
    \draw (p8) -- (p9);
  \end{scope}
  \begin{scope}[xshift=5*50]
    \examplePoints{$\sts$}
    \draw (p0) -- (p1);
    \draw (p1) -- (p4);
    \draw (p4) -- (p5);
    \draw (p5) -- (p8);
    \draw (p2) -- (p3);
    \draw (p3) -- (p5);
    \draw (p3) -- (p6);
    \draw (p7) -- (p8);
    \draw (p7) -- (p9);
  \end{scope}
  \begin{scope}[xshift=6*50]
    \examplePoints{$\scs$}
    \draw (p0) -- (p1);
    \draw (p1) -- (p4);
    \draw (p4) -- (p5);
    \draw (p5) -- (p8);
    \draw (p8) -- (p9);
    \draw (p0) -- (p2);
    \draw (p2) -- (p3);
    \draw (p3) -- (p7);
    \draw (p7) -- (p6);
    \draw (p6) -- (p9);
  \end{scope}
\end{tikzpicture}

%% file: tab/extremalbounds.tex
\newcommand{\citebox}[1]{\hbox to 10pt{\cite{#1}}}
\newcommand{\citeboxnote}[1]{\hbox to 10pt{\cite{#1}\footnotemark}}
\newcommand{\nocitebox}{\hbox to 10pt{}}
\newcommand{\ns}{\hspace{-5pt}}

\footnotesize
\begin{tabular}{lrrrrrrr}
  \toprule
  & \ns$\pgsnum$
  & \ns$\cpsnum$                   & \ns$\pmsnum$
  & \ns$\cssnum$                   & \ns$\trsnum$
  & \ns$\stsnum$                   & \ns$\scsnum$                \\
  \midrule
    \ns$\forall \pts \colon O^*(c^n)$
  & \ns$187.53$ \citebox{SS12}
  & \ns$12.24$  \citebox{SW06p}    & \ns$10.05$ \citebox{SW06p}
  & \ns$187.53$ \citebox{SS12}     & \ns$30.00$ \citebox{SS09}
  & \ns$141.07$ \citebox{HSSTW11}  & \ns$54.55$ \citebox{SSW11}  \\
    \ns$\exists \pts \colon \Omega^*(c^n)$
  & \ns$41.18$  \citebox{AHHHKV07}
  & \ns$5.23$   \citebox{SW06p}    & \ns$3.00$  \citebox{GNT00}
  & \ns$8.65$   \citebox{DSST10}   & \ns$8.65$  \citebox{DSST10}
  & \ns$12.52$  \citebox{HM15}   & \ns$4.64$  \citebox{GNT00}  \\
    \ns$\forall \pts \colon \Omega^*(c^n)$
  & \ns$11.65$  \citebox{FN99}
  & \ns$3.00$   \citeboxnote{GNT00}& \ns$2.00$  \citebox{GNT00} 
  & \ns$2.43$   \citebox{SSW10}    & \ns$2.43$  \citebox{SSW10}   
  & \ns$6.75$   \citebox{FN99}     & \ns$1.00$  \nocitebox       \\
  \bottomrule
\end{tabular}

%% file: basics.tex
\section{Definitions and Results}\label{sec:basics}

Let $\frags$ be the set of all crossing-free geometric graphs on all non-empty subsets of $\pts$.
Elements of $\frags$ are called \emph{units}, and they can be thought of as the simple ``pieces'' from the informal introduction.
For every unit $\unit \in \frags$ we denote by $\upts{\unit} \subseteq \pts$ the set of vertex points of $\unit$.
That is, if $\unit$ is a geometric graph on a particular subset $\pts'$ of $\pts$, then $\upts{\unit} = \pts'$.

Let us now define a number of useful subsets of $\frags$.
\begin{itemize}
\item Let $\segs$ be the set of \emph{segments} with both endpoints in $\pts$.
That is, each $\unit \in \segs$ is a geometric graph on exactly two points of $\pts$ with the edge between them.
\item Let $\cparts$ be the set of \emph{convex parts} with vertex points in $\pts$.
That is, for each $\unit \in \cparts$ the convex hull of $\upts{\unit}$ does not contain any points of $\pts \setminus \upts{\unit}$ (in words, interior points are also vertex points of $\unit$).
Moreover, $\unit$ contains all the edges along the boundary of the convex hull of the set $\upts{\unit}$.
Observe that all isolated points and segments with endpoints in $\pts$ are also elements of $\cparts$.
\item Let $\cpartstp$ be the set of \emph{convex faces} with vertex points in $\pts$. That is, $\cpartstp$ contains all $\unit \in \cparts$ with $|\upts{\unit}| \geq 3$ and no interior vertices.
The \emph{shape} of any such $\unit$ is the bounded and closed region delimited by its edges.
\item Let $\trgs$ be the set of \emph{empty triangles} with vertex points in $\pts$.
That is, $\trgs$ contains all $\unit \in \cpartstp$ with $|\upts{\unit}| = 3$.
\end{itemize}

Let us fix a subset $\units$ of $\frags$.
For any units $\unit_1,\unit_2 \in \units$, if $\unit_1$ contains an edge $e_1$ and $\unit_2$ contains a different edge $e_2$ such that $e_1$ and $e_2$ are crossing, then we also say that $\unit_1$ and $\unit_2$ are \emph{crossing}.
Otherwise, $\unit_1$ and $\unit_2$ are \emph{non-crossing}.
A \emph{combination} of $\units$ is a subset $\cmb$ of $\units$, and $\cmb$ is \emph{crossing-free} if the elements of $\cmb$ are pairwise non-crossing.

We denote by $\cmbscf{\units}$ the set of all crossing-free combinations of $\units$.
Furthermore, a combination $\cmb \in \cmbscf{\units}$ is called a \emph{partition} of $\pts$ if the sets $\upts{\unit}$ of all $\unit \in \cmb$ are pairwise disjoint and if their union is equal to $\pts$.
We denote by $\cmbspt{\units}$ the set of all such partitions of $\pts$.
Assuming $\units \subseteq \cpartstp$, a combination $\cmb \in \cmbscf{\units}$ is called a \emph{subdivision} of the convex hull of $\pts$ if the shapes of all $\unit \in \cmb$ are pairwise interior-disjoint and if their union is equal to the convex hull of $\pts$.
We denote by $\cmbssd{\units}$ the set of all such subdivisions of the convex hull of $\pts$.
\Cref{fig:combinations} shows one representative of each type of crossing-free combinations for the special case $\units = \trgs$.
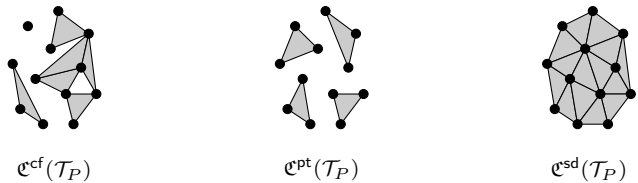
\begin{figure}
  \begin{center}
    \input{fig/combinations}
  \end{center}
  \caption{
    The three defined types of crossing-free combinations of $\trgs$.
    The shaded regions represent the shapes of individual elements of $\trgs$.
  }
  \label{fig:combinations}
\end{figure}

From now on we will no longer consider the sets $\pgs$, $\cps$, and so on, as defined in the introduction.
Instead, we will talk about crossing-free combinations of specific sets of units.
For example, since there is an obvious bijection between the sets $\cmbscf{\segs}$ and $\pgs$, any counting or enumeration algorithm for one set can be adapted easily for the other.
Similarly, there are bijections between the sets $\cmbspt{\cparts}$ and $\cps$, $\cmbspt{\segs}$ and $\pms$, $\cmbssd{\cpartstp}$ and $\css$, as well as $\cmbssd{\trgs}$ and $\trs$.

In the same spirit, we define the sets $\cmbsst,\cmbssc \subseteq \cmbscf{\segs}$ of all crossing-free combinations of $\segs$ whose segments form spanning trees and spanning cycles on $\pts$, respectively.

\begin{definition}\label{def:cgraph}
  Let $\units \subseteq \frags$.
  A \emph{combination graph (over $\units$)} is a directed and acyclic multigraph $\cgraph$ with two distinguished vertices $\cgsource$ and $\cgsink$, called the \emph{source} and \emph{sink} of $\cgraph$, respectively.
  All edges in $\cgraph$, except for those ending in $\cgsink$, are labeled with an element of $\units$.
  Moreover, the sink $\cgsink$ has no outgoing edges.
  The \emph{size} of $\cgraph$ is the number of vertices and edges in $\cgraph$.
\end{definition}

For any combination graph $\cgraph$ and any set $\cmbs$ of combinations, we say that $\cgraph$ \emph{represents} $\cmbs$ if there is a bijection between the set of directed $\cgsource$-$\cgsink$ paths in $\cgraph$ and the set $\cmbs$ in the following sense.
Taking any $\cgsource$-$\cgsink$ path in $\cgraph$ and building the set of labels on that path yields the corresponding combination in $\cmbs$.

The following are comparatively simple applications of an abstract framework developed in \cref{sec:framework}.
The corresponding proofs can be found in \cref{sec:applications}.
Many more, in some cases obvious, applications are possible.
\begin{theorem}[All geometric graphs]\label{thm:pg}
  There exists a combination graph over $\segs$ of size $O(c^nn^3)$ with $c < 2.83929$ that represents $\cmbscf{\segs}$.
\end{theorem}
\begin{theorem}[Convex partitions]\label{thm:cp}
  There exists a combination graph over $\cparts$ of size $O(2^nn^3)$ that represents $\cmbspt{\cparts}$.
\end{theorem}
\begin{theorem}[Perfect matchings]\label{thm:pm}
  There exists a combination graph over $\segs$ of size $O(2^nn^3)$ that represents $\cmbspt{\segs}$.
\end{theorem}
\begin{theorem}[Convex subdivisions]\label{thm:cs}
  There exists a combination graph over $\cpartstp$ of size $O(2^nn^3)$ that represents $\cmbssd{\cpartstp}$.
\end{theorem}
\begin{theorem}[Triangulations; c.f.\ Theorem 3 in \cite{SA13}]\label{thm:tr}
  There exists a combination graph over $\trgs$ of size $O(2^nn^3)$ that represents $\cmbssd{\trgs}$.
\end{theorem}

Within our framework we can explain similar results for spanning trees and spanning cycles.
However, these two classes are substantially harder to deal with.
\Cref{sec:spanning} is devoted to the corresponding proofs.
\begin{theorem}[Spanning trees]\label{thm:st}
  There exists a combination graph over $\segs$ of size $O(c^nn^3)$ with $c < 7.04313$ that represents $\cmbsst$.
\end{theorem}
\begin{theorem}[Spanning cycles]\label{thm:sc}
  There exists a combination graph over $\segs$ of size $O(c^nn^3)$ with $c < 5.61804$ that represents $\cmbssc$.
\end{theorem}

To get a bound on the running time for computing an explicit representation of any one of the combination graphs $\cgraph$ in the above theorems, it suffices to add another factor $\ptsnum$ to the bound on the size of $\cgraph$.
See \cref{sec:construction} for details.

Given such a representation of $\cgraph$, the corresponding counting problem can be solved in time linear in the size of $\cgraph$ by counting directed $\cgsource$-$\cgsink$ paths using standard graph algorithms.
After removing all dead ends in $\cgraph$, which is also possible in time linear in the size of $\cgraph$, enumeration of the corresponding set $\cmbs$ requires time at most linear in the length of the longest $\cgsource$-$\cgsink$ path per enumerated object.
We will abstain from describing the required algorithms in detail, and instead refer to \cite{SA13} for an example.

Observe that the exponential bases in \Cref{thm:pg,thm:cp,thm:pm,thm:cs,thm:tr} are not larger than the exponential bases of the corresponding lower bounds given in the third row of \Cref{tab:extremalbounds}.
As a corollary we therefore get enumeration algorithms for the sets $\pgs$, $\cps$, $\pms$, $\css$ and $\trs$ whose overall running times are bounded by the length of the output times a polynomial in $n$.
For $\cps$ and $\pms$, a small adaptation, which is described in \cref{sec:enumeration}, results in enumeration algorithms with polynomial delay for each (in particular, the first) output.

\begin{restatable}{theorem}{thmenum}
  The sets $\cps$ and $\pms$ can be enumerated such that the time delay for any output is bounded by a polynomial in $n$.
\end{restatable}

With the exception of $\pms$, the lower bounds are even strictly larger, which means that our algorithms compute the numbers $\pgsnum$, $\cpsnum$, $\cssnum$ and $\trsnum$ with exponential speed-up over any procedures that count by enumerating the respective sets.
For spanning trees it might as well be that the constant $c$ in \cref{thm:st} is smaller than $6.75$, but we were unable to prove it.
For spanning cycles we cannot always hope for such an exponential speed-up because for a set $\pts$ of $\ptsnum$ points in convex position we have $\scsnum = 1$.

%% file: fig/combinations.tex
\begin{tikzpicture}
  \newcommand\exampleCoords{
    \coordinate (c0) at (0.0,1.0);
    \coordinate (c1) at (0.1,0.4);
    \coordinate (c2) at (0.2,1.5);
    \coordinate (c3) at (0.3,0.8);
    \coordinate (c4) at (0.4,0.2);
    \coordinate (c5) at (0.5,1.2);
    \coordinate (c6) at (0.6,1.7);
    \coordinate (c7) at (0.7,0.6);
    \coordinate (c8) at (0.8,0.2);
    \coordinate (c9) at (0.9,0.95);
    \coordinate (c10) at (1.0,1.4);
    \coordinate (c11) at (1.1,0.6);
  }
  \newcommand\examplePoints[1]{
    \node[point] (p0) at (c0) {};
    \node[point] (p1) at (c1) {};
    \node[point] (p2) at (c2) {};
    \node[point] (p3) at (c3) {};
    \node[point] (p4) at (c4) {};
    \node[point] (p5) at (c5) {};
    \node[point] (p6) at (c6) {};
    \node[point] (p7) at (c7) {};
    \node[point] (p8) at (c8) {};
    \node[point] (p9) at (c9) {};
    \node[point] (p10) at (c10) {};
    \node[point] (p11) at (c11) {};
    \node (label) at (0.55,-0.4) {\footnotesize #1};
  }
  \begin{scope}[xshift=0]
    \exampleCoords
    \draw[fill=black!20!white] (c0) -- (c1) -- (c4) -- cycle;
    \draw[fill=black!20!white] (c3) -- (c7) -- (c9) -- cycle;
    \draw[fill=black!20!white] (c3) -- (c9) -- (c10) -- cycle;
    \draw[fill=black!20!white] (c5) -- (c6) -- (c10) -- cycle;
    \draw[fill=black!20!white] (c7) -- (c8) -- (c11) -- cycle;
    \draw[fill=black!20!white] (c9) -- (c10) -- (c11) -- cycle;
    \examplePoints{$\cmbscf{\trgs}$}
  \end{scope}
  \begin{scope}[xshift=100]
    \exampleCoords
    \draw[fill=black!20!white] (c0) -- (c2) -- (c5) -- cycle;
    \draw[fill=black!20!white] (c1) -- (c3) -- (c4) -- cycle;
    \draw[fill=black!20!white] (c6) -- (c9) -- (c10) -- cycle;
    \draw[fill=black!20!white] (c7) -- (c8) -- (c11) -- cycle;
    \examplePoints{$\cmbspt{\trgs}$}
  \end{scope}
  \begin{scope}[xshift=200]
    \exampleCoords
    \draw[fill=black!20!white] (c0) -- (c1) -- (c3) -- cycle;
    \draw[fill=black!20!white] (c0) -- (c2) -- (c5) -- cycle;
    \draw[fill=black!20!white] (c0) -- (c3) -- (c5) -- cycle;
    \draw[fill=black!20!white] (c1) -- (c3) -- (c4) -- cycle;
    \draw[fill=black!20!white] (c3) -- (c4) -- (c7) -- cycle;
    \draw[fill=black!20!white] (c3) -- (c5) -- (c7) -- cycle;
    \draw[fill=black!20!white] (c5) -- (c7) -- (c9) -- cycle;
    \draw[fill=black!20!white] (c7) -- (c9) -- (c11) -- cycle;
    \draw[fill=black!20!white] (c7) -- (c8) -- (c11) -- cycle;
    \draw[fill=black!20!white] (c4) -- (c7) -- (c8) -- cycle;
    \draw[fill=black!20!white] (c2) -- (c5) -- (c6) -- cycle;
    \draw[fill=black!20!white] (c5) -- (c6) -- (c10) -- cycle;
    \draw[fill=black!20!white] (c5) -- (c9) -- (c10) -- cycle;
    \draw[fill=black!20!white] (c9) -- (c10) -- (c11) -- cycle;
    \examplePoints{$\cmbssd{\trgs}$}
  \end{scope}
\end{tikzpicture}

%% file: framework.tex
\section{The Abstract Framework}\label{sec:framework}

As before, let $\pts$ be a set of $\ptsnum$ points in general position.
Furthermore, assume that no two points share the same $x$-coordinate, which means the points can be ordered as $\pt_1,\dots,\pt_\ptsnum$ from left to right in a unique way.
If a point $\pt_i$ is to the left of another point $\pt_j$, that is, if $i \leq j$, then we write $\pt_i \ptspreceq \pt_j$.
Recall that $\upts{\unit} \subseteq \pts$ is the set of vertex points in any $\unit \in \frags$.
We define $\ulft{\unit} \coloneqq \min_{\ptspreceq}(\upts{\unit})$ and $\urgt{\unit} \coloneqq \max_{\ptspreceq}(\upts{\unit})$, the \emph{left-most} and \emph{right-most point} of $\unit$, respectively.
For any $\unit_1,\unit_2 \in \frags$, if $\urgt{\unit_1} \ptspreceq \ulft{\unit_2}$ holds then we say that $\unit_1$ is \emph{to the left} of $\unit_2$ and we also write $\unit_1 \utorgt \unit_2$.

For each $\unit$ we define $\ulow{\unit} \subseteq \pts$ and $\uupp{\unit} \subseteq \pts$, the \emph{lower} and \emph{upper shadow} of $\unit$, respectively.
The set $\ulow{\unit}$ contains all points in $\pts$ from which a vertical ray shooting upwards intersects the relative interior of some segment of $\unit$.
The set $\uupp{\unit}$ is defined in an analogous way.
Whenever we have either $\upts{\unit_1} \cap \ulow{\unit_2} \neq \emptyset$ or $\uupp{\unit_1} \cap \upts{\unit_2} \neq \emptyset$ for any $\unit_1,\unit_2$, then we say that \emph{$\unit_2$ depends on $\unit_1$} and we write $\unit_1 \udep \unit_2$.
The following lemma aims at making this cryptic definition more intuitive.

Here, and only here, by a \emph{point on $\unit$} we mean either a point in $\upts{\unit}$ or a point in the relative interior of some edge of $\unit$.

\begin{lemma}\label{obs:dependence}
  Let $\unit_1,\unit_2 \in \frags$ be arbitrary.
  Then, $\unit_2$ depends on $\unit_1$ if and only if there exists a point on $\unit_2$ directly (that is, same $x$-coordinate) and strictly above a point on $\unit_1$.
  In particular, if $\unit_1$ and $\unit_2$ are crossing then they are mutually dependent, that is, $\unit_1 \udep \unit_2$ and $\unit_2 \udep \unit_1$.
\end{lemma}
\hspace*{-\parindent}%
\begin{minipage}{0.8\textwidth}
\begin{proof}
    The ``only if'' is immediate by definition of $\unit_1 \udep \unit_2$.
    For the ``if'', let us fix two points on $\unit_1$ and $\unit_2$, respectively, the one on $\unit_2$ directly and strictly above the one on $\unit_1$.
    If either of those points is contained in $\upts{\unit_1}$ or $\upts{\unit_2}$, respectively, the conclusion $\unit_1 \udep \unit_2$ is again immediate.
    Otherwise, let $e_1$ and $e_2$ be the edges of $\unit_1$ and $\unit_2$ that contain the two respective fixed points in their relative interiors.
    Without loss of generality we assume that $e_1$ and $e_2$ diverge and thus do not intersect towards the left, as illustrated on the right.
    In case 1, the left endpoint of $e_1$ is below $e_2$, which means $\upts{\unit_1} \cap \ulow{\unit_2} \neq \emptyset$ and hence $\unit_1 \udep \unit_2$.
    In case 2, the left endpoint of $e_2$ is above $e_1$, which means $\uupp{\unit_1} \cap \upts{\unit_2} \neq \emptyset$ and hence again $\unit_1 \udep \unit_2$.
\end{proof}
\end{minipage}
\begin{minipage}{0.2\textwidth}
  \centering
  \input{fig/dependingsegments}
\end{minipage}
\vspace{\belowdisplayskip}

To ease notation throughout, we use $\cpts{\cmb}$ and $\clow{\cmb}$ to denote the sets $\bigcup_{\unit \in \cmb} \upts{\unit}$ and $\bigcup_{\unit \in \cmb} \ulow{\unit}$, respectively, where $\cmb$ is any combination.

Besides giving an intuition for the dependence relation, the above lemma turns out to be absolutely crucial for everything that follows.
It suggests a safe and practical way of adding a new segment $\unit \in \segs$, say, to a crossing-free combination $\cmb$ of $\segs$.
Safe means that we do not introduce any crossings, that is, $\cmb \cup \{\unit\}$ is itself crossing-free.
Practical means that we may, to a great extent, remain ignorant of the exact composition of $\cmb$.
Indeed, as long as we know the sets $\cpts{\cmb}$ and $\clow{\cmb}$, and provided that we choose $\unit$ such that $\upts{\unit} \cap \clow{\cmb} = \emptyset$ and $\uupp{\unit} \cap \cpts{\cmb} = \emptyset$, then no element of $\cmb$ can possibly depend on $\unit$ and hence, by \cref{obs:dependence}, $\cmb \cup \{\unit\}$ is crossing-free.

By extension, this suggests a way of constructing a combination $\cmb$ by adding all elements in a succession where earlier occurrences do not depend on later occurrences.
An apparent disadvantage is that this will not work for every conceivable subset $\units$ of $\frags$ and every crossing-free combination $\cmb$ of $\units$.
Most importantly, for it to work, there must be no circular dependencies among elements of $\cmb$.
In the following we formalize this requirement.

Let $\units \subseteq \frags$ and let $\cmb \subseteq \units$ be arbitrary.
An element $\unit$ of $\cmb$ is \emph{extreme (in $\cmb$)} if $\unit \undep \unit'$ holds for all other elements $\unit'$ of $\cmb$.
If it exists, the \emph{right-most extreme element} in $\cmb$ is the unique extreme element $\unit$ in $\cmb$ which satisfies $\unit' \utorgt \unit$ for all other extreme elements $\unit'$ in $\cmb$.

\begin{definition}\label{def:serializable2}
  Let $\units \subseteq \frags$ and let $\cmbs$ be a set of combinations of $\units$.
  We call $\cmbs$ \emph{serializable} if it is non-empty and if every non-empty $\cmb \in \cmbs$ contains a right-most extreme element, which we then denote by $\umaxin{\cmb}$.
  Additionally, $\cmb \setminus \{\umaxin{\cmb}\}$ must itself be an element of $\cmbs$.
\end{definition}

Let $\cmbs$ be a serializable set of combinations of $\units$, and let $\cmb,\cmb' \in \cmbs$.
We will often write $\ctarrow{\cmb}{\cmb'}{\unit}$, which stands for $\cmb = \cmb' \setminus \{\unit\}$ and $\unit = \umaxin{\cmb'}$.
Observe that $\cmbs$ naturally induces a directed and acyclic graph (actually, a tree) with vertex set $\cmbs$ and edges with labels from the set $\units$.
Indeed, whenever $\ctarrow{\cmb}{\cmb'}{\unit}$ holds we simply add an edge from vertex $\cmb$ to vertex $\cmb'$ with label $\unit$.
A combination graph over $\units$ that represents an arbitrary subset of $\cmbs$ is obtained by defining $\cgsource \coloneqq \emptyset$ and by adding appropriate unlabeled edges which end in an additional vertex $\cgsink$.
However, the resulting combination graph is too large as its size is $\Theta(|\cmbs|)$.
In the following we show how to compress it without losing too many of its nice properties.

\begin{definition}\label{def:coherent}
  Let $\cmbs$ be serializable.
  An equivalence relation $\parteq$ on $\cmbs$ is called \emph{coherent} if $\cmb_1 \parteq \cmb_2$ implies the following.
If $\cmb_1 \xrightarrow{\unit} \cmb_1'$ holds then also $\cmb_2 \xrightarrow{\unit} \cmb_2'$ for some $\cmb_2' \parteq \cmb_1'$. 
\end{definition}

Intuitively, to make our combination graph smaller we would like to merge two vertices $\cmb_1$ and $\cmb_2$.
This makes sense only if the subtrees rooted at $\cmb_1$ and $\cmb_2$ are identical when looking at edge labels.
As will be proved later, if $\cmb_1 \parteq \cmb_2$ holds, coherency of $\parteq$ enforces precisely what we want.

In the remainder of the paper we will always deal with a serializable set $\cmbs$ of combinations of some $\units \subseteq \frags$, and an equivalence relation $\parteq$ on $\cmbs$.
For any $\cmb \in \cmbs$ we then define the equivalence class $\partclass{\cmb} \coloneqq \{\cmb' \in \cmbs \colon \cmb' \parteq \cmb\}$, where the relation $\parteq$ will be obvious from the context.
We also define the set $\quotient{\cmbs}{\parteq} \coloneqq \{\partclass{\cmb} \colon \cmb \in \cmbs\}$ of all equivalence classes.

\begin{definition}
  A \emph{combination problem (on $\pts$)} is a tuple $(\units,\cmbs,\parteq,\target)$ where $\units \subseteq \frags$, $\cmbs$ is a serializable set of combinations of $\units$, $\parteq$ is a coherent equivalence relation on $\cmbs$, and $\target$ is a subset of $\quotient{\cmbs}{\parteq}$.
\end{definition}

As will be described below, the set $\target$ allows us to specify the subset of $\cmbs$ we are actually interested in.

Every combination problem $(\units,\cmbs,\parteq,\target)$ induces a corresponding combination graph $\cgraph = \cgraph(\units,\cmbs,\parteq,\target)$ over $\units$ as follows.
The vertices of $\cgraph$ are all equivalence classes in $\quotient{\cmbs}{\parteq}$ plus one extra vertex denoted by $\cgsink$.
The vertex $\partclass{\emptyset}$ is also referred to by $\cgsource$.
Existence of $\cgsource$ follows from $\emptyset \in \cmbs$, an easy consequence of serializability of $\cmbs$.
Whenever $\ctarrow{\cmb}{\cmb'}{\unit}$ holds, then we add an edge from vertex $\partclass{\cmb}$ to vertex $\partclass{\cmb'}$ with label $\unit$.
We do not, however, add the same labeled edge more than once.
Lastly, for every vertex $\partclass{\cmb}$ in $\target$ we add an unlabeled edge which starts in $\partclass{\cmb}$ and ends in $\cgsink$.
Observe that $\cgraph$ does not contain any directed cycles because given any such cycle, it would be possible to construct an infinite sequence $(\ctarrow{\cmb}{\cmb'}{\unit},\ctarrow{\cmb'}{\cmb''}{\unit'},\dots)$, which cannot exist.

Recall that $\ctarrow{\cmb}{\cmb'}{\unit}$ stands for $\cmb = \cmb' \setminus \{\unit\}$ and $\unit = \umaxin{\cmb'}$.
Additionally, we write $\cgarrow{\cmb}{\cmb'}{\unit}$ if there exists an edge in $\cgraph$ from vertex $\partclass{\cmb}$ to vertex $\partclass{\cmb'}$ with label $\unit$.
The following observations are straight-forward consequences of coherency of $\parteq$ and of the way $\cgraph$ is constructed.

\begin{observation}\label{obs:cgproperties}
  The combination graph $\cgraph = \cgraph(\units,\cmbs,\parteq,\target)$, defined as above, is
  \begin{itemize}
    \item
      Complete: if $\ctarrow{\cmb}{\cmb'}{\unit}$, then
      $\cgarrow{\cmb}{\cmb'}{\unit}$,
    \item
      Sound: if $\cgarrow{\cmb}{\cmb'}{\unit}$, then $\ctarrow{\cmb}{\cmb''}{\unit}$
      for some $\cmb'' \parteq \cmb'$,
    \item
      Deterministic: if $\cgarrow{\cmb}{\cmb'}{\unit}$ and $\cgarrow{\cmb}{\cmb''}{\unit}$,
      then $\partclass{\cmb'} = \partclass{\cmb''}$.
  \end{itemize}
\end{observation}

By induction it can now be shown that for every vertex $\partclass{\cmb}$ there is a natural bijection from $\partclass{\cmb}$ to the set of directed $\cgsource$-$\partclass{\cmb}$ paths in $\cgraph$, which implies the following lemma.

\begin{restatable}{lemma}{represent}
  \label{lem:represent}
  Let $(\units,\cmbs,\parteq,\target)$ be a combination problem and let $\cgraph$ be the corresponding combination graph.
  Then, $\cgraph$ represents $\bigcup\target$ and the size of $\cgraph$ is at most $O(|\quotient{\cmbs}{\parteq}| \cdot |\units|)$.
\end{restatable}

\begin{proof}
  The upper bound on the size of $\cgraph$ holds because it has exactly $|\quotient{\cmbs}{\parteq}| + 1$ vertices and, since it is deterministic (as in \Cref{obs:cgproperties}) and has no duplicate labeled edges, each vertex has at most $|\units| + 1$ outgoing edges.
  It only remains to show that $\cgraph$ represents $\bigcup\target$.
  
  For each directed path $\cgpath$ in $\cgraph$ let $\cmbntr{\cgpath}$ denote the set of all labels on $\cgpath$, that is, $\cmbntr{\cgpath}$ is the combination of $\units$ corresponding to the path $\cgpath$.
  Moreover, for any $\cmb \in \cmbs$, the \emph{canonical order} of $\cmb$ is the unique sequence over all elements in $\cmb$ which is obtained by successively removing right-most extreme elements from $\cmb$ and then putting the removed elements in reverse order.

  For every vertex $\partclass{\cmb}$ in $\cgraph$ we now claim the following.
  Firstly, there exists a directed $\cgsource$-$\partclass{\cmb}$ path $\cgpath$ in $\cgraph$ with $\cmbntr{\cgpath} = \cmb$, irrespective of the choice of the representative $\cmb$.
  Secondly, for every directed $\cgsource$-$\partclass{\cmb}$ path $\cgpath$ we have that $\cmbntr{\cgpath} \in \partclass{\cmb}$.
  Thirdly, the labels on any directed $\cgsource$-$\partclass{\cmb}$ path $\cgpath$ with $\cmbntr{\cgpath} = \cmb$ appear in canonical order of $\cmb$.
  The proof of these claims is by induction over an arbitrary topological ordering of the vertices in $\cgraph$.
  After that, the lemma follows by combining the claims with the fact that $\cgraph$ is deterministic.

  For the first part of the claim, if $\cmb = \emptyset$, then $\partclass{\cmb} = \cgsource$ and the $\cgsource$-$\cgsource$ path without edges works.
  Otherwise, if $\cmb \neq \emptyset$, then we know that $\ctarrow{\cmb'}{\cmb}{\unit}$ for some $\cmb'$ and $\unit$.
  By completeness, there is an edge $\cgarrow{\cmb'}{\cmb}{\unit}$ in $\cgraph$ and a $\cgsource$-$\partclass{\cmb}$ path $\cgpath$ with $\cmbntr{\cgpath} = \cmb$ can be constructed from the $\cgsource$-$\partclass{\cmb'}$ path $\cgpath'$ with $\cmbntr{\cgpath'} = \cmb'$, which exists by induction.

  For the second part of the claim, let $\cgpath$ be a directed $\cgsource$-$\partclass{\cmb}$ path in $\cgraph$.
  If $\cgpath$ is of length zero then $\partclass{\cmb} = \cgsource$ and $\cmbntr{\cgpath} = \emptyset \in \partclass{\cmb}$.
  Otherwise, let $\cgpath'$ be the $\cgsource$-$\partclass{\cmb'}$ path that is obtained by removing the last edge $\cgarrow{\cmb'}{\cmb}{\unit}$ from $\cgpath$.
  By induction, and without loss of generality, $\cmbntr{\cgpath'} = \cmb'$.
  By soundness, we have $\ctarrow{\cmb'}{\cmb''}{\unit}$ for some $\cmb'' \parteq \cmb$ and hence also $\cmbntr{\cgpath} = \cmb'' \in \partclass{\cmb}$.

  For the third part of the claim, simply observe that in the previous paragraph, the last label on the path $\cgpath$ with $\cmbntr{\cgpath} = \cmb''$ is the right-most extreme element in $\cmb''$.
  Hence, this last claim also follows by induction. 
\end{proof}

%% file: fig/dependingsegments.tex
\begin{tikzpicture}
  \begin{scope}[yshift=70]
    \draw [dashed] (0,0) -- ++(190:0.5) coordinate [label=below:{$e_1$}] (i1);
    \draw [solid]  (i1) -- ++(190:0.666) node [point] (e1) {};
    \coordinate (s2) at (0,0.3);
    \draw [dashed] (s2) -- ++(170:0.5) coordinate [label=above:{$e_2$}] (i2);
    \draw [solid]  (i2) -- ++(170:1) node [point] (e2) {};
    \draw [dotted] (e2) -- ++(270:1.0);
    \node (h) at (-0.6,-0.9) {\footnotesize case 1};
  \end{scope}
  \begin{scope}[yshift=0]
    \draw [dashed] (0,0) -- ++(190:0.5) coordinate [label=below:{$e_1$}] (i1);
    \draw [solid]  (i1) -- ++(190:1) node [point] (e1) {};
    \draw [dotted] (e1) -- ++(90:1.0);
    \coordinate (s2) at (0,0.3);
    \draw [dashed] (s2) -- ++(170:0.5) coordinate [label=above:{$e_2$}] (i2);
    \draw [solid]  (i2) -- ++(170:0.666) node [point] (e2) {};
    \node (h) at (-0.6,-0.9) {\footnotesize case 2};
  \end{scope}
\end{tikzpicture}

%% file: applications.tex
\section{Three Simple Applications}\label{sec:applications}

We present three generic kinds of combination problems.
They directly correspond to the three types of crossing-free combinations depicted earlier in \cref{fig:combinations}.

\subsection{All Crossing-free Combinations}\label{subsec:crossingfree}
The main aim of this subsection is to give a proof of \cref{thm:pg}.
Still, some of the following insights are fairly general and can be used in many different settings.

Let us fix a set of units $\units \subseteq \frags$.
We make the explicit assumption that $\cmbscf{\units}$, the set of crossing-free combinations of $\units$, is serializable.
Then, as follows, a combination $\cmb$ in $\cmbscf{\units}$ can be described by a coloring of the point set $\pts$ with three colors $\ptfree,\ptalive,\ptdead$ and a special marking, e.g., $\ptaliveM$, on one of the points.
A point $\pt \in \pts$ is given the color $\ptdead$ if $\pt \in \clow{\cmb}$, it is given the color $\ptalive$ if $\pt \in \cpts{\cmb} \setminus \clow{\cmb}$, and it is given the color $\ptfree$ in all other cases.
The special marking is put on the left-most point of $\umaxin{\cmb}$.
\Cref{fig:encoding_crossingfree} shows, for the special case $\units = \segs$, that different crossing-free combinations can have identical such descriptions.
Whenever that is the case, we consider two combinations equivalent.
More formally, we put $\cmb \parteq \cmb'$ if and only if\footnote{Of course, the third condition only makes sense if $\cmb$ and $\cmb'$ are non-empty.}
\label{def:cf:eq}
\begin{itemize}
  \item $\clow{\cmb} = \clow{\cmb'}$,
  \item $\cpts{\cmb} \setminus \clow{\cmb} = \cpts{\cmb'} \setminus \clow{\cmb'}$,
  \item $\ulft{\umaxin{\cmb}} = \ulft{\umaxin{\cmb'}}$.
\end{itemize}

\begin{figure}[b]
  \begin{center}
    \input{fig/encoding_crossingfree}
  \end{center}
  \caption{
    Three elements of $\cmbscf{\segs}$ that are considered equivalent.
  }
  \label{fig:encoding_crossingfree}
\end{figure}

The goal now is to prove that $\parteq$ is coherent, as in \cref{def:coherent}.
Unfortunately, this endeavor is doomed to fail because for some contrived choices of $\units$ we can have $\ctarrow{\cmb}{\cmb'}{\unit}$ and, at the same time, $\cmb \parteq \cmb'$. 
In the language of combination graphs this means that we would have to introduce loops, which leads to a potentially infinite number of source-sink paths.
To avoid this problem, we require for any $\cmb,\cmb' \in \cmbscf{\units}$ and any $\unit \in \units$ that $\ctarrow{\cmb}{\cmb'}{\unit}$ implies $\cmb \not\parteq \cmb'$.
If this additional requirement is met, then we say that $\cmbscf{\units}$ is \emph{progressive}.

\begin{lemma}\label{lem:cf:coherent}
  Let $\units \subseteq \frags$ with $\cmbscf{\units}$ both serializable and progressive.
  Then, the equivalence relation $\parteq$ on $\cmbscf{\units}$, as defined above, is coherent.
\end{lemma}
\begin{proof}
  Let $\cmb_1,\cmb_2 \in \cmbscf{\units}$ be non-empty (otherwise, the proof is trivial) with $\cmb_1 \parteq \cmb_2$ and assume that $\ctarrow{\cmb_1}{\cmb_1'}{\unit}$ holds for $\cmb_1' \in \cmbscf{\units}$ and $\unit \in \units$.
  Consider $\cmb_2' \coloneqq \cmb_2 \cup \{\unit\}$.
  We show $\cmb_2' \in \cmbscf{\units}$, $\cmb_1' \parteq \cmb_2'$ and $\ctarrow{\cmb_2}{\cmb_2'}{\unit}$,  which implies the lemma.
  
  By assumption we have $\clow{\cmb_1} = \clow{\cmb_2}$ and $\cpts{\cmb_1} \setminus \clow{\cmb_1} = \cpts{\cmb_2} \setminus \clow{\cmb_2}$.
  From this and from the fact that $\unit$ is extreme in $\cmb_1'$, we first derive that $\unit$ is also extreme in $\cmb_2'$.
  So, assume the opposite, that is, either $\upts{\unit} \cap \clow{\cmb_2} \neq \emptyset$ or $\uupp{\unit} \cap \cpts{\cmb_2} \neq \emptyset$ holds.
  In the first case, we get $\upts{\unit} \cap \clow{\cmb_1} = \upts{\unit} \cap \clow{\cmb_2} \neq \emptyset$, contradicting that $\unit$ is extreme in $\cmb_1'$.
  In the second case, we get at least one of $\uupp{\unit} \cap \cpts{\cmb_1} \supseteq \uupp{\unit} \cap (\cpts{\cmb_1} \setminus \clow{\cmb_1}) = \uupp{\unit} \cap (\cpts{\cmb_2} \setminus \clow{\cmb_2}) \neq \emptyset$ and $\uupp{\unit} \cap \clow{\cmb_1} = \uupp{\unit} \cap \clow{\cmb_2} \supseteq \uupp{\unit} \cap (\cpts{\cmb_2} \cap \clow{\cmb_2}) \neq \emptyset$.
  Both possibilities again lead to $\unit$ not being extreme in $\cmb_1'$.
  We conclude that, indeed, $\unit$ is extreme in $\cmb_2'$.

  Since $\unit$ is extreme in $\cmb_2'$, no other element of $\cmb_2'$ depends on $\unit$.
  By \cref{obs:dependence}, since $\cmb_2$ is crossing-free, also $\cmb_2'$ must be crossing-free.
  Hence, $\cmb_2' \in \cmbscf{\units}$.

  Next, by definition of $\cmb_2'$ it is easily seen that $\clow{\cmb_1'} = \clow{\cmb_2'}$ and $\cpts{\cmb_1'} \setminus \clow{\cmb_1'} = \cpts{\cmb_2'} \setminus \clow{\cmb_2'}$ both hold.
  To prove also $\ulft{\umaxin{\cmb_1'}} = \ulft{\umaxin{\cmb_2'}}$, and consequently $\cmb_1' \parteq \cmb_2'$, it suffices to show $\umaxin{\cmb_2'} = \unit$. 
  So, for $\unit' \coloneqq \umaxin{\cmb_2'}$ and assuming that $\unit' \neq \unit$, we get $\unit \utorgt \unit'$ and thus also $\umaxin{\cmb_2} = \unit'$.
  Invoking $\cmb_1 \parteq \cmb_2$ yields $\unit \utorgt \umaxin{\cmb_1}$, which contradicts the assumption $\unit = \umaxin{\cmb_1'}$.

  It remains to prove $\ctarrow{\cmb_2}{\cmb_2'}{\unit}$.
  The crucial observation here is that $\unit = \umaxin{\cmb_2'}$ and $\cmb_2 \cup \{\unit\} = \cmb_2'$ alone are not sufficient.
  What we need instead is $\cmb_2 = \cmb_2' \setminus \{\unit\}$, that is, $\unit$ must not be contained in $\cmb_2$.
  However, assuming that $\unit$ is contained in $\cmb_2$, we easily derive $\cmb_1 \parteq \cmb_2 = \cmb_2' \parteq \cmb_1'$, which, when combined with $\ctarrow{\cmb_1}{\cmb_1'}{\unit}$, contradicts the fact that $\cmbscf{\units}$ is progressive.
\end{proof}

We define $\target \coloneqq \quotient{\cmbscf{\units}}{\parteq}$.
With \cref{lem:cf:coherent} we see that $(\units,\cmbscf{\units},\parteq,\target)$
is a combination problem as long as $\cmbscf{\units}$ is serializable and progressive.
The corresponding combination graph over $\units$ is referred to by
$\agraphcf{\units}$.

\begin{corollary}\label{cor:crossingfree}
  If $\units \subseteq \frags$ is such that $\cmbscf{\units}$ is both serializable and progressive, then $\agraphcf{\units}$ represents $\cmbscf{\units}$ and the size of $\agraphcf{\units}$ is at most $O(|\quotient{\cmbscf{\units}}{\parteq}| \cdot |\units|)$.
\end{corollary}

\Cref{thm:pg} follows by invoking \cref{cor:crossingfree} with $\units = \segs$ and by making use of the next
two lemmas. 
The first one essentially shows that $\cmbscf{\segs}$ is a well-behaved set.
The second one gives us a better upper bound on the size of \mbox{$\agraphcf{\segs}$.
Observe,} however, that a bound of $O(3^nn^3)$ is immediate because we can encode equivalence classes with 3 colors and with one marking, and because of $|\segs| = O(\ptsnum^2)$.

\begin{restatable}{lemma}{cfserializable}
  \label{lem:cf:serializable}
  For any point set $\pts$, $\cmbscf{\segs}$ is both serializable and progressive.
\end{restatable}


\begin{proof}[Proof (serializable)]
  Let $\cmb \in \cmbscf{\segs}$ be arbitrary but non-empty.
  We prove serializability by exhibiting a right-most extreme element $\unit$ in $\cmb$.
  For this, consider the relative interiors of all segments in $\cmb$.
  These segments without endpoints are convex and pairwise non-intersecting.
  For a set of convex and non-intersecting shapes in the plane it is well known that at least one of them can be translated in the $y$-direction to infinity in a continuous motion, without intersecting any other segment in the process \cite{T85}.
  Any segment with this property corresponds to an extreme element in $\cmb$.
  The extreme elements in $\cmb$ now can be ordered from left to right according to their relative positions when projected orthogonally onto the $x$-axis.
  The segment $\unit$ on the far right is a right-most extreme element in $\cmb$.
\end{proof}
\hspace*{-\parindent}%
\begin{minipage}{0.8\textwidth}
  \begin{proof}[Proof (progressive)]
    For the sake of contradiction, assume that we have $\ctarrow{\cmb'}{\cmb}{\unit}$ and $\cmb \parteq \cmb'$ for some $\cmb,\cmb' \in \cmbscf{\segs}$ and $\unit \in \segs$. 
    Let $\unit' \coloneqq \umaxin{\cmb'}$ (note that the case where $\cmb'$ is empty is trivial).
    We have $\ulft{\unit} = \ulft{\unit'}$ and we have to distinguish the two cases illustrated on the right.
    In case 1, we see that $\urgt{\unit} \in \cpts{\cmb} \setminus \clow{\cmb}$
    but also $\urgt{\unit} \not\in \cpts{\cmb'} \setminus \clow{\cmb'}$ since
    otherwise $\unit'$ would not be extreme in $\cmb'$.
    In case 2, we see that $\urgt{\unit'} \in \clow{\cmb}$ but also
    $\urgt{\unit'} \not\in \clow{\cmb'}$ since otherwise $\unit'$ would not be
    extreme in $\cmb'$.
    Hence, in both cases we get a contradiction to the assumption $\cmb \parteq \cmb'$.
  \end{proof}
\end{minipage}
\begin{minipage}{0.2\textwidth}
  \centering
  \input{fig/progressive}
\end{minipage}
\vspace{\belowdisplayskip}

\begin{restatable}{lemma}{cfbound}
  \label{lem:cf:bound}
  For any set $\pts$ of $\ptsnum$ points, the relation $\parteq$ partitions $\cmbscf{\segs}$ into at most $O(\alpha^nn)$ equivalence classes, that is, $|\quotient{\cmbscf{\segs}}{\parteq}| = O(\alpha^nn)$, where $\alpha \lessapprox 2.83929$.\footnote{We write $\alpha \lessapprox \beta$ if a parameter $\alpha$ is approximately equal to $\beta$ and also strictly smaller.}
\end{restatable}

The proof of \cref{lem:cf:bound} is a bit tedious.
We only sketch the main idea here, and postpone a more careful analysis to \cref{sec:bounds}.

Note that for any three consecutive points $\pt_i,\pt_{i+1},\pt_{i+2}$ in $\pts$, the point $\pt_{i+1}$ is either below or above the straight line through $\pt_i$ and $\pt_{i+2}$, as depicted in \cref{fig:upperbound_crossingfree}.
\begin{figure}[b]
  \begin{center}
    \input{fig/upperbound_crossingfree}
  \end{center}
  \caption{
    Three consecutive points can always be colored in a way that does not correspond to an element of $\cmbscf{\segs}$.
  }
  \label{fig:upperbound_crossingfree}
\end{figure}
In both cases we can show that at least one of the $3^3 = 27$ different ways of assigning colors to $\pt_i,\pt_{i+1},\pt_{i+2}$ does not describe an actual element of $\cmbscf{\segs}$.
A bound of $O(26^{n/3}n) = O(2.963^nn)$ on the size of $\quotient{\cmbscf{\segs}}{\parteq}$ then follows after partitioning $\pts$ into $n/3$ consecutive triples.

In the first case, if we assign colors $\ptdead,\ptfree,\ptdead$, there must be two distinct segments which pass over the points $\pt_i$ and $\pt_{i+2}$, respectively, and which pass under $\pt_{i+1}$.
Clearly, any two such segments are crossing.
In the second case, if we assign colors $\ptfree,\ptdead,\ptfree$, there must be a segment that passes over $\pt_{i+1}$ and under $\pt_i$ and $\pt_{i+2}$.
Clearly, only a bent ``segment'' can achieve that.

Note also that if the combination graph $\agraphcf{\segs}$ is constructed bottom-up, as will be explained in more detail in \cref{sec:construction}, then all these impossible colorings are avoided automatically.


\subsection{Crossing-free Partitions}\label{subsec:partitions}

\begin{figure}[b]
  \begin{center}
    \input{fig/encoding_partition}
  \end{center}
  \caption{
    Three elements of $\protect\cmbslpt{\cparts}$ that are considered equivalent.
  }
  \label{fig:encoding_partition}
\end{figure}

Let us again fix a set of units $\units \subseteq \frags$.
We define the set $\cmbslpt{\units}$ of all crossing-free combinations $\cmb$ of $\units$
for which the sets $\upts{\unit}$ of all $\unit \in \cmb$ are pairwise disjoint and
for which $\clow{\cmb} \subseteq \cpts{\cmb}$ holds.
\Cref{fig:encoding_partition} depicts three of these combinations for the special case $\units = \cparts$.
Observe that $\cmbspt{\units}$ and $\cmbslpt{\units}$ are different sets, and also observe that we have in fact $\cmbspt{\units} \subseteq \cmbslpt{\units}$.

Assume that $\cmbslpt{\units}$ is serializable.
Similar to the previous subsection, we use two colors $\ptfree,\ptalive$ and a special marking $\ptaliveM$ on the points in $\pts$ to describe an element $\cmb$ of $\cmbslpt{\units}$.
A point $\pt$ receives the color $\ptalive$ if $\pt \in \cpts{\cmb}$, and $\ptfree$ otherwise.
The marking is again put on the left-most point of $\umaxin{\cmb}$. 
If two combinations have identical such descriptions, we consider them equivalent.
Formally, we put $\cmb \parteq \cmb'$ if and only if
\begin{itemize}
  \item $\cpts{\cmb} = \cpts{\cmb'}$,
  \item $\ulft{\umaxin{\cmb}} = \ulft{\umaxin{\cmb'}}$.
\end{itemize}

One peculiarity in \cref{fig:encoding_partition} is that some points have been given the color $\ptalive$ even though there are no incident segments.
This is because the set $\cparts$ also contains all isolated points, that is, all $\unit \in \frags$ with $|\upts{\unit}| = 1$.

Also observe that $\ctarrow{\cmb}{\cmb'}{\unit}$ implies $\cmb \not\parteq \cmb'$ because $\cpts{\cmb}$ is always a proper subset of $\cpts{\cmb'}$.
Consequently, an explicit notion of ``progressive'', as in the previous subsection, is not needed here.

\begin{lemma}\label{lem:pt:coherent}
  Let $\units \subseteq \frags$ with $\cmbslpt{\units}$ serializable.
  Then, the equivalence relation $\parteq$ on $\cmbslpt{\units}$, as defined above, is coherent.
\end{lemma}
\begin{proof}
  Let $\cmb_1,\cmb_2 \in \cmbslpt{\units}$ be non-empty (otherwise, the proof is trivial) with $\cmb_1 \parteq \cmb_2$ and assume that $\ctarrow{\cmb_1}{\cmb_1'}{\unit}$ holds for $\cmb_1' \in \cmbslpt{\units}$ and $\unit \in \units$.
  Consider $\cmb_2' \coloneqq \cmb_2 \cup \{\unit\}$.
  We show $\cmb_2' \in \cmbslpt{\units}$, $\cmb_1' \parteq \cmb_2'$ and $\ctarrow{\cmb_2}{\cmb_2'}{\unit}$, which implies the lemma.
  
  First, observe that since $\unit$ is extreme in $\cmb_1'$, it is also extreme in $\cmb_2'$.
  Indeed, $\upts{\unit} \cap \clow{\cmb_2} \subseteq \upts{\unit} \cap \cpts{\cmb_2} = \upts{\unit} \cap \cpts{\cmb_1} = \emptyset$ and $\uupp{\unit} \cap \cpts{\cmb_2} = \uupp{\unit} \cap \cpts{\cmb_1} = \emptyset$ can be derived immediately.
  By \cref{obs:dependence}, $\cmb_2'$ is crossing-free.
  Moreover, deriving $\clow{\cmb_2'} = \clow{\cmb_2} \cup \ulow{\unit} \subseteq \cpts{\cmb_2} \cup \clow{\cmb_1'} \subseteq \cpts{\cmb_2'} \cup \cpts{\cmb_1'} = \cpts{\cmb_2'}$ proves that, indeed, $\cmb_2' \in \cmbslpt{\units}$.
  By making use of the marking in the same way as in the proof of \cref{lem:cf:coherent} we get $\umaxin{\cmb_2'} = \unit$, and it follows that $\cmb_1' \parteq \cmb_2'$ and $\ctarrow{\cmb_2}{\cmb_2'}{\unit}$.
\end{proof}

We define the set $\target \subseteq \quotient{\cmbslpt{\units}}{\parteq}$
which contains all equivalence classes $\partclass{\cmb}$ for which
$\cpts{\cmb} = \pts$ holds.
Observe that then $\bigcup \target = \cmbspt{\units}$.
From \Cref{lem:pt:coherent} it follows that $(\units,\cmbslpt{\units},\parteq,\target)$ is a combination
problem provided that $\cmbslpt{\units}$ is serializable.
We denote by $\agraphpt{\units}$ the corresponding combination graph.

\begin{corollary}\label{cor:partitions}
  If $\units \subseteq \frags$ is such that $\cmbslpt{\units}$ is serializable, then $\agraphpt{\units}$ represents $\cmbspt{\units}$ and the size of $\agraphpt{\units}$ is at most $O(|\quotient{\cmbslpt{\units}}{\parteq}| \cdot |\units|)$. 
\end{corollary}

\Cref{thm:cp,thm:pm} now follow from \cref{cor:partitions} and the following
two lemmas.

\begin{restatable}{lemma}{ptserializable}
  For any point set $\pts$ and any subset $\units$ of $\cparts$, it holds that $\cmbslpt{\units}$ is serializable.
\end{restatable}
\begin{proof}
  The proof for the existence of right-most extreme elements is analogous to the first part of the proof of \Cref{lem:cf:serializable}.
  However, proving that $\cmb' \coloneqq \cmb \setminus \{\unit\}$, where $\unit = \umaxin{\cmb}$, is also an element of $\cmbslpt{\units}$ is not completely trivial since we have to verify that $\clow{\cmb'} \subseteq \cpts{\cmb'}$ holds.

  For the sake of contradiction, let us assume that there exists a point $\pt \in \clow{\cmb'} \setminus \cpts{\cmb'}$.
  From $\cmb' \subseteq \cmb$ we get $\clow{\cmb'} \subseteq \clow{\cmb}$ and hence also $\pt \in \clow{\cmb}$.
  We now use $\clow{\cmb} \subseteq \cpts{\cmb}$, which holds by definition of $\cmbslpt{\units}$, to obtain $\pt \in \cpts{\cmb}$.
  By combining this with the assumption $\pt \not\in \cpts{\cmb'}$ we obtain $\pt \in \upts{\unit}$ because $\cpts{\cmb} \setminus \cpts{\cmb'} = \upts{\unit}$.
  It follows that $\clow{\cmb'} \cap \upts{\unit}$ is non-empty since it contains at least the point $\pt$.
  This, however, contradicts the fact that $\unit$ is extreme in $\cmb$.
\end{proof}

Note that in the case of convex partitions, the general bound on the size of the resulting combination graph from \cref{cor:partitions} is insufficient to prove \Cref{thm:cp} because $\units = \cparts$ alone can be of size $\Omega(2^\ptsnum)$.
The following lemma is therefore really needed.

\begin{restatable}{lemma}{ptbound}
  \label{lem:pt:bound}
  For any point set $\pts$ of size $\ptsnum$ and any subset $\units$ of $\cparts$, the size of $\agraphpt{\units}$ is at most $O(2^nn^3)$.
\end{restatable}
\begin{proof}
  We prove that the number of labeled edges in $\agraphpt{\units}$ can be bounded by $O(2^nn^3)$ from above, which then implies the lemma.

  Fix $\pt_l, \pt_r \in \pts$ with $\pt_l \ptspreceq \pt_r$, and let $\pts_l^r \coloneqq \{\pt_l,\pt_{l+1},\dots,\pt_{r-1},\pt_r\} \subseteq \pts$ be the set of points between $\pt_l$ and $\pt_r$.
  Observe that there are at most $2^{r-l+1}$ convex parts $\unit \in \cparts$ for which $\pt_l = \ulft{\unit}$ and $\pt_r = \urgt{\unit}$ holds.
  Let us also fix such a convex part $\unit$. 
  Next, we will give a bound on the number of edges in $\agraphpt{\units}$ with label $\unit$.
  This bound will only depend on the indices $l$ and $r$.
  
  Suppose that $\cgarrow{\cmb}{\cmb'}{\unit}$ is an edge in $\agraphpt{\units}$ and assume further that $\ctarrow{\cmb}{\cmb'}{\unit}$ holds.
  Then, we easily see that $\uslow{\unit} = \cpts{\cmb} \cap \pts_l^r$, where $\uslow{\unit} \coloneqq \ulow{\unit} \setminus \upts{\unit}$.
  This means that for each point $\pt \in \pts_l^r$ it is determined by $\unit$ whether $\pt \in \cpts{\cmb}$ holds or not.
  For a point $\pt \in \pts \setminus \pts_l^r$ there are at most two choices, either $\pt \in \cpts{\cmb}$ or $\pt \not\in \cpts{\cmb}$.
  It follows that there are at most $2^{n-(r-l+1)}n$ many vertices $\partclass{\cmb}$ with an outgoing edge that is labeled by $\unit$.
  As usual, the additional factor $n$ comes from the special marking.

  The total number of labeled edges in $\agraphpt{\units}$ can therefore be bounded by
  \begin{align*}
    \sum_{l=1}^{n} \sum_{r=l}^{n} 2^{r-l+1} \cdot 2^{n-(r-l+1)}n = \sum_{l=1}^{n} \sum_{r=l}^{n} 2^nn = O(2^nn^3)\text{.}
  \end{align*}
\end{proof}


\subsection{Subdivisions}
Let $\units \subseteq \cpartstp$, which means in particular that the \emph{shape} of each unit $\unit$ is defined.
We define the set $\cmbslsd{\units}$ which contains all combinations $\cmb$ of $\units$ for which the following holds.
There exists an $x$-monotone polygonal chain, denoted by $\cchn{\cmb}$, which starts in $\pt_1$, ends in $\pt_\ptsnum$, has only points from $\pts$ as vertices, and satisfies the following with regard to $\cmb$.
The shapes of all $\unit$ in $\cmb$ form a subdivision of the region between $\cchn{\cmb}$ and the lower convex hull of $\pts$, by which we mean that the shapes are pairwise interior-disjoint and each point of the plane in the interior of that region is contained in the shape of at least one element $\unit$ of $\cmb$.

As depicted in \cref{fig:encoding_subdivision}, we describe such a combination $\cmb$ by giving the vertex points of $\cchn{\cmb}$ the color $\ptalive$, by giving all other points the color $\ptfree$, and by adding the usual marking.
Guided by this description, we put $\cmb \parteq \cmb'$ if and only if
\begin{itemize}
  \item $\cchn{\cmb} = \cchn{\cmb'}$,
  \item $\ulft{\umaxin{\cmb}} = \ulft{\umaxin{\cmb'}}$.
\end{itemize}
\begin{figure}
  \begin{center}
    \input{fig/encoding_subdivision}
  \end{center}
  \caption{
    Three elements of $\protect\cmbslsd{\cpartstp}$ that are considered equivalent.
  }
  \label{fig:encoding_subdivision}
\end{figure}

\begin{lemma}
  \label{lem:sd:coherent}
  Let $\units \subseteq \cpartstp$ with $\cmbslsd{\units}$ serializable.
  Then, the equivalence relation $\parteq$ on $\cmbslsd{\units}$, as defined above, is coherent.
\end{lemma}
\begin{proof}
  Let $\cmb_1,\cmb_2 \in \cmbslsd{\units}$ be non-empty (otherwise, the proof is trivial) with $\cmb_1 \parteq \cmb_2$ and assume that $\ctarrow{\cmb_1}{\cmb_1'}{\unit}$ holds for $\cmb_1' \in \cmbslsd{\units}$ and $\unit \in \units$.
  Consider $\cmb_2' \coloneqq \cmb_2 \cup \{\unit\}$.
  We show $\cmb_2' \in \cmbslsd{\units}$, $\cmb_1' \parteq \cmb_2'$ and $\ctarrow{\cmb_2}{\cmb_2'}{\unit}$, which implies the lemma.
  
  From the definition of $\cmbslsd{\units}$, it is immediate that also $\cmb_2' \in \cmbslsd{\units}$ and $\cchn{\cmb_2'} = \cchn{\cmb_1'}$.
  Since $\unit$ is extreme in $\cmb_1'$, it is also extreme in $\cmb_2'$.
  Using the usual argument involving the marking we see that $\umaxin{\cmb_2'} = \unit$, it follows that $\cmb_1' \parteq \cmb_2'$ and $\ctarrow{\cmb_2}{\cmb_2'}{\unit}$.
\end{proof}

Let $\target \subseteq \quotient{\cmbslsd{\units}}{\parteq}$ be the set that contains $\partclass{\cmb}$ if and only if $\cchn{\cmb}$ is equal to the upper convex hull of $\pts$.
Observe that $\bigcup \target = \cmbssd{\units}$.
If $\units$ is such that $\cmbslsd{\units}$ is serializable, then $(\units,\cmbslsd{\units},\parteq,\target)$ is a combination problem and we denote by $\agraphsd{\units}$ the corresponding combination graph.

\begin{corollary}\label{cor:subdivisions}
  If $\units \subseteq \cpartstp$ is such that $\cmbslsd{\units}$ is serializable, then $\agraphsd{\units}$ represents $\cmbssd{\units}$ and the size of $\agraphsd{\units}$ is at most $O(|\quotient{\cmbslsd{\units}}{\parteq}| \cdot |\units|)$. 
\end{corollary}

\Cref{thm:cs,thm:tr} now follow from \cref{cor:subdivisions} and the following two lemmas.

\begin{restatable}{lemma}{sdserializable}
  For any point set $\pts$ and any subset $\units$ of $\cpartstp$, it holds that $\cmbslsd{\units}$ is serializable.
\end{restatable}
\begin{proof}
  The proof for the existence of right-most extreme elements is again analogous to \Cref{lem:cf:serializable}.
  Moreover, it is not hard to see that $\cmb' \coloneqq \cmb \setminus \{\unit\}$, where $\unit = \umaxin{\cmb}$, is also an element of $\cmbslsd{\units}$.
  Simply observe that the upper convex hull of the shape of $\unit$ must be contained in $\cchn{\cmb}$, which means that $\cchn{\cmb'}$ is obtained from $\cchn{\cmb}$ by replacing the upper hull of $\unit$ with its lower hull.
\end{proof}

Similar to the previous subsection, the next lemma is not implied by the general bound from \cref{cor:subdivisions}.

\begin{restatable}{lemma}{sdbound}
  \label{lem:sd:bound}
  For any point set $\pts$ of size $\ptsnum$ and any subset $\units$ of $\cpartstp$, the size of $\agraphsd{\units}$ is at most $O(2^nn^3)$.
\end{restatable}
\begin{proof}
  Analogous to the proof of \cref{lem:pt:bound}.
\end{proof}

%% file: fig/encoding_crossingfree.tex
\begin{tikzpicture}
  \newcommand\exampleCoords{
    \coordinate (c0) at (0.0,1.0);
    \coordinate (c1) at (0.15,0.4);
    \coordinate (c2) at (0.3,1.7);
    \coordinate (c3) at (0.45,0.9);
    \coordinate (c4) at (0.6,1.3);
    \coordinate (c5) at (0.75,0.1);
    \coordinate (c6) at (0.9,1.9);
    \coordinate (c7) at (1.05,0.9);
    \coordinate (c8) at (1.2,0.3);
    \coordinate (c9) at (1.35,1.2);
    \coordinate (c10) at (1.5,0.8);
  }
  \newcommand\examplePoints{
    \node[alive] (p0) at (c0) {};
    \node[dead] (p1) at (c1) {};
    \node[free] (p2) at (c2) {};
    \node[free] (p3) at (c3) {};
    \node[alive] (p4) at (c4) {};
    \node[marked] (m) at (c4) {};
    \node[dead] (p5) at (c5) {};
    \node[free] (p6) at (c6) {};
    \node[dead] (p7) at (c7) {};
    \node[dead] (p8) at (c8) {};
    \node[alive] (p9) at (c9) {};
    \node[alive] (p10) at (c10) {};
  }
  \newcommand\exampleLabel[1]{
    \node (label) at (0.55,-0.4) {\footnotesize #1};
  }
  \begin{scope}[xshift=0]
    \exampleCoords
    \draw (c0) -- (c5);
    \draw (c4) -- (c9);
    \draw (c8) -- (c10);
    \examplePoints
  \end{scope}
  \begin{scope}[xshift=100]
    \exampleCoords
    \draw (c0) -- (c8);
    \draw (c1) -- (c8);
    \draw (c4) -- (c10);
    \draw (c4) -- (c9);
    \draw (c4) -- (c8);
    \draw (c5) -- (c8);
    \examplePoints
  \end{scope}
  \begin{scope}[xshift=200]
    \exampleCoords
    \draw (c0) -- (c1);
    \draw (c0) -- (c5);
    \draw (c1) -- (c5);
    \draw (c4) -- (c5);
    \draw (c4) -- (c9);
    \draw (c5) -- (c9);
    \draw (c8) -- (c10);
    \examplePoints
  \end{scope}
\end{tikzpicture}

%% file: fig/progressive.tex
\begin{tikzpicture}
  \begin{scope}[yshift=0]
    \node[point] (l) at (0,0) {};
    \draw (l) -- ++(10:1.0) node[point] (r1) {} node[midway,above] {$\unit$};
    \draw (l) -- ++(-10:1.5) node[point] (r2) {} node[midway,below] {$\unit'$};
    \draw [dotted] (r2) -- ++(0,0.8);
    \node (h) at (0.74,-0.9) {\footnotesize case 1};
  \end{scope}
  \begin{scope}[yshift=-60]
    \node[point] (l) at (0,0) {};
    \draw (l) -- ++(10:1.5) node[point] (r1) {} node[midway,above] {$\unit$};
    \draw [label={[below]{$\unit'$}}] (l) -- ++(-10:1.0) node[point] (r2) {} node[midway,below] {$\unit'$};
    \draw [dotted] (r1) -- ++(0,-0.8);
    \node (h) at (0.74,-0.9) {\footnotesize case 2};
  \end{scope}
\end{tikzpicture}

%% file: fig/upperbound_crossingfree.tex
\begin{tikzpicture}
  \begin{scope}[xshift=0]
    \node[dead] (p1) at (0,0.33) {};
    \node[free] (p2) at (1,0) {};
    \node[dead] (p3) at (2,0.33) {};
    \coordinate (a1) at ($(p1)+(0,0.2)$);
    \coordinate (b2) at ($(p2)-(0,0.2)$);
    \coordinate (a3) at ($(p3)+(0,0.2)$);
    \draw[shorten <= -0.25cm, shorten >= -0.5cm] (a1) -- (b2);
    \draw[shorten <= -0.5cm, shorten >= -0.25cm] (b2) -- (a3); 
  \end{scope}
  \begin{scope}[xshift=150]
    \node[free] (p1) at (0,0) {};
    \node[dead] (p2) at (1,0.33) {};
    \node[free] (p3) at (2,0) {};
    \coordinate (b1) at ($(p1)-(0,0.2)$);
    \coordinate (a2) at ($(p2)+(0,0.2)$);
    \coordinate (b3) at ($(p3)-(0,0.2)$);
    \draw[rounded corners] (b1) [shorten <= -0.25cm] -- (a2) -- (b3) [shorten >= -0.25cm]; 
  \end{scope}
\end{tikzpicture}

%% file: fig/encoding_partition.tex
\begin{tikzpicture}
  \newcommand\exampleCoords{
    \coordinate (c0) at (0.0,1.0);
    \coordinate (c1) at (0.15,0.4);
    \coordinate (c2) at (0.3,1.7);
    \coordinate (c3) at (0.45,0.9);
    \coordinate (c4) at (0.6,1.3);
    \coordinate (c5) at (0.75,0.1);
    \coordinate (c6) at (0.9,1.9);
    \coordinate (c7) at (1.05,0.9);
    \coordinate (c8) at (1.2,0.3);
    \coordinate (c9) at (1.35,1.2);
    \coordinate (c10) at (1.5,0.8);
  }
  \newcommand\examplePoints{
    \node[alive] (p0) at (c0) {};
    \node[alive] (p1) at (c1) {};
    \node[alive] (p2) at (c2) {};
    \node[free] (p3) at (c3) {};
    \node[alive] (p4) at (c4) {};
    \node[marked] (m) at (c4) {};
    \node[alive] (p5) at (c5) {};
    \node[free] (p6) at (c6) {};
    \node[alive] (p7) at (c7) {};
    \node[alive] (p8) at (c8) {};
    \node[free] (p9) at (c9) {};
    \node[alive] (p10) at (c10) {};
  }
  \newcommand\exampleLabel[1]{
    \node (label) at (0.55,-0.4) {\footnotesize #1};
  }
  \begin{scope}[xshift=0]
    \exampleCoords
    \draw (c0) -- (c2);
    \draw (c1) -- (c5);
    \draw (c4) -- (c10);
    \draw (c7) -- (c8);
    \examplePoints
  \end{scope}
  \begin{scope}[xshift=100]
    \exampleCoords
    \draw[fill=black!20!white] (c0) -- (c1) -- (c5) -- cycle;
    \draw[fill=black!20!white] (c4) -- (c8) -- (c10) -- cycle;
    \examplePoints
  \end{scope}
  \begin{scope}[xshift=200]
    \exampleCoords
    \draw (c0) -- (c2);
    \draw[fill=black!20!white] (c1) -- (c10) -- (c7) -- cycle;
    \draw (c5) -- (c8);
    \examplePoints
  \end{scope}
\end{tikzpicture}

%% file: fig/encoding_subdivision.tex
\begin{tikzpicture}
  \newcommand\exampleCoords{
    \coordinate (c0) at (0.0,1.0);
    \coordinate (c1) at (0.15,0.4);
    \coordinate (c2) at (0.3,1.7);
    \coordinate (c3) at (0.45,0.9);
    \coordinate (c4) at (0.6,1.3);
    \coordinate (c5) at (0.75,0.1);
    \coordinate (c6) at (0.9,1.9);
    \coordinate (c7) at (1.05,0.9);
    \coordinate (c8) at (1.2,0.3);
    \coordinate (c9) at (1.35,1.2);
    \coordinate (c10) at (1.5,0.8);
  }
  \newcommand\examplePoints{
    \node[alive] (p0) at (c0) {};
    \node[free] (p1) at (c1) {};
    \node[free] (p2) at (c2) {};
    \node[alive] (p3) at (c3) {};
    \node[alive] (p4) at (c4) {};
    \node[alive] (p4) at (c4) {};
    \node[marked] (m) at (c4) {};
    \node[free] (p5) at (c5) {};
    \node[free] (p6) at (c6) {};
    \node[free] (p7) at (c7) {};
    \node[free] (p8) at (c8) {};
    \node[alive] (p9) at (c9) {};
    \node[alive] (p10) at (c10) {};
  }
  \newcommand\exampleLabel[1]{
    \node (label) at (0.55,-0.4) {\footnotesize #1};
  }
  \begin{scope}[xshift=0]
    \exampleCoords
    \draw[fill=black!20!white] (c0) -- (c1) -- (c3) -- cycle;
    \draw[fill=black!20!white] (c1) -- (c3) -- (c7) -- cycle;
    \draw[fill=black!20!white] (c3) -- (c4) -- (c7) -- cycle;
    \draw[fill=black!20!white] (c1) -- (c5) -- (c7) -- cycle;
    \draw[fill=black!20!white] (c5) -- (c7) -- (c8) -- cycle;
    \draw[fill=black!20!white] (c4) -- (c7) -- (c9) -- cycle;
    \draw[fill=black!20!white] (c7) -- (c9) -- (c10) -- cycle;
    \draw[fill=black!20!white] (c7) -- (c8) -- (c10) -- cycle;
    \examplePoints
  \end{scope}
  \begin{scope}[xshift=200]
    \exampleCoords
    \draw[fill=black!20!white] (c0) -- (c1) -- (c3) -- cycle;
    \draw[fill=black!20!white] (c1) -- (c3) -- (c8) -- (c5) -- cycle;
    \draw[fill=black!20!white] (c3) -- (c4) -- (c8) -- cycle;
    \draw[fill=black!20!white] (c4) -- (c8) -- (c9) -- cycle;
    \draw[fill=black!20!white] (c4) -- (c8) -- (c7) -- cycle;
    \draw[fill=black!20!white] (c4) -- (c7) -- (c9) -- cycle;
    \draw[fill=black!20!white] (c7) -- (c8) -- (c9) -- cycle;
    \draw[fill=black!20!white] (c8) -- (c9) -- (c10) -- cycle;
    \examplePoints
  \end{scope}
  \begin{scope}[xshift=100]
    \exampleCoords
    \draw[fill=black!20!white] (c0) -- (c1) -- (c5) -- (c3) -- cycle;
    \draw[fill=black!20!white] (c3) -- (c5) -- (c8) -- (c7) -- (c4) -- cycle;
    \draw[fill=black!20!white] (c7) -- (c8) -- (c10) -- cycle;
    \draw[fill=black!20!white] (c4) -- (c9) -- (c10) -- (c7) -- cycle;
    \examplePoints
  \end{scope}
\end{tikzpicture}

%% file: construction.tex
\section{Representations of Combination Graphs}\label{sec:construction}

In the preceding sections, vertices of a combination graph $\cgraph$ were always treated as equivalence classes over a set of combinations.
Of course, this is a very inefficient way to represent them in an actual algorithm.
However, in all cases we have seen how to describe these equivalence classes by an assignment of a constant number of colors to the points in $\pts$ and an index to the point with the special marking.
It is thus easy to encode vertices of $\cgraph$ with a linear number of bits for the coloring and a logarithmic number of bits for the index.

As for the construction of $\cgraph$, this is most easily done bottom-up.
That is, we start with the source $\cgsource$, we enumerate all outgoing edges and add the corresponding new vertices to our representation of $\cgraph$.
We continue this process, that is, pick a vertex and enumerate all outgoing edges, until we have done so for all vertices in $\cgraph$.

In all cases except for convex partitions and convex subdivisions, the enumeration of outgoing edges for a given vertex can be done in the most wasteful way while still staying within the required time bounds.
For example, in the case of \cref{thm:pg}, for every vertex in $\agraphcf{\segs}$ we can simply enumerate the whole set $\segs$ and check for each segment whether it corresponds to an outgoing edge in linear time.
In the case of \cref{thm:cp,thm:cs} we have to be more careful since we do not have the time to enumerate the whole set $\cparts$, say, for every vertex in $\agraphpt{\cparts}$.
A simple solution is to enumerate the set $\cparts$ only once at the beginning and to find for each enumerated $\unit \in \cparts$ all vertices in $\agraphpt{\cparts}$ for which $\unit$ corresponds to an outgoing edge.
The running time of this approach can be bounded by using the same arguments as in the proof of \cref{lem:pt:bound}.

%% file: enumeration.tex
\section{Enumeration with Polynomial-Time Delay}\label{sec:enumeration}

In this section we sketch an additional trick\footnote{This trick was suggested to the author by Emo Welzl.}, which allows us to enumerate the sets of crossing-free convex partitions and perfect matchings with polynomial-time delay.
That is, we give algorithms which output the elements of $\cps$ or $\pms$ in some order and without repetitions, and such that the time we have to wait for any new output is not larger than a polynomial.

As already discussed in \cref{sec:basics}, from \cref{thm:cp,thm:pm} we get enumeration algorithms for these two sets that run in time $O^\ast(\cpsnum)$ and $O^\ast(\pmsnum)$, respectively.
However, while the time delay \emph{between} any two outputs is bounded by a polyomial, the time delay \emph{before} the first output is exponential.
Precisely, there is a preprocessing phase that might take time $\Theta^\ast(2^n)$ during which we construct representations of the respective combination graphs $\agraphpt{\cparts}$ and $\agraphpt{\segs}$ and during which we do not produce any outputs.
The trick is to hide this preprocessing phase by outputting $\Theta^\ast(2^n)$ objects obtained by other means.

\thmenum*

\begin{proof}

We begin by defining a sufficiently large subset $\pmss$ of $\cps$ and $\pms$.
Elements of this set are called \emph{easy perfect matchings}, and they are constructed recursively.
If $\pts$ is the empty set, then $\pmss = \pms$.
Otherwise, let $\pt_1$ be the left-most point and let $\pt_i$ be any other point.
Let $\pts_i$ be the set of points that are to the left of the directed line through $\pt_1$ and $\pt_i$, and let $\pts_i'$ be the set of points that are to the right of that line.
The set $\pmss$ contains all perfect matchings that, for any choice of $i$, are composed of the edge $\pt_1\pt_i$ and two easy perfect matchings on $P_i$ and $P_i'$, respectively.

It is clear that easy perfect matchings are crossing-free.
Furthermore, note that efficient enumeration and recognition algorithms for the set $\pmss$ are easy to obtain from the definition.
Lastly, the number of easy perfect matchings on even-sized point sets satisfies the Catalan recurrence, and thus $|\pmss| = \Theta(2^n/n^{3/2})$.
Refer to \cite{GNT00} for more details, where these objects are used to prove a corresponding lower bound on the number of perfect matchings.

We now have everything that we need.
In order to enumerate the set $\pms$, say, with polynomial-time delay, we start the construction of $\agraphpt{\segs}$.
During this preprocessing phase, we output elements of $\pmss$ in appropriate time intervals.
Once we have an explicit representation of $\agraphpt{\segs}$, we continue the enumeration by outputting arbitrary elements of $\pms$.
Of course, whenever we have a new potential output, we have to check first whether it is an easy perfect matching, which means that it has been output before.
If that is the case, then we simply discard it.
One final caveat now is that there might be a long period where we have to discard all potential outputs, which might again lead to a delay that is no longer polynomially bounded.
However, this is easily fixed for example by only using up half of the set $\pmss$ during the preprocessing phase, and by using the other half as a substitute for every other discarded output during the second phase.
\end{proof}

%% file: spanning.tex
\section{Spanning Trees and Spanning Cycles}\label{sec:spanning}

In this section we show that it is possible to construct non-trivial combination graphs for the sets of crossing-free spanning trees and spanning cycles.

Spanning trees and spanning cycles are harder to deal with than anything that we have encountered before.
The reason is that these graphs have properties which hold globally.
For example, the construction of $\agraphpt{\segs}$ in \cref{subsec:partitions} can be adapted in such a way that source-sink paths correspond to 2-regular (instead of 1-regular) crossing-free geometric graphs.
We simply would have to keep track of the degree of each vertex (whether it is currently 0, 1, or 2, which means we would need three instead of two colors) and in the end require that every vertex has degree 2.
However, if we want that source-sink paths correspond only to crossing-free spanning cycles, then we also need that each such path corresponds to a connected geometric graph.
Being connected is such a property that holds globally, and there seems to be no obvious and efficient way to deal with it.
To get rid of this problem, at least in the case of spanning trees and spanning cycles, we next state an auxiliary lemma.
It will allow us to translate connectivity into simpler features which can be enforced on a local level.

Let $G$ be a directed multigraph\footnote{To avoid confusion, let us stress that we introduce completely new entities here. That is, $G$ is neither a geometric graph nor a combination graph.} and let $v$ be a vertex in $G$.
$G$ is \emph{root-oriented towards $v$} if all vertices in $G$ have exactly one outgoing edge, except for $v$, which has no outgoing edges.
If $G$ is root-oriented towards $v$, then $v$ is called the \emph{root} of $G$.
Observe that being root-oriented implies that $G$ has exactly $n-1$ edges, where $n$ is the number of vertices in $G$.
It does however not imply that $G$ is connected or, in other words, a tree.
The reason is that there might be a connected component with a directed cycle.
Such components are always disconnected from the root.

A \emph{plane drawing} of $G$ is a drawing which maps all vertices of $G$ to distinct points in the plane and which draws all edges as simple curves such that no two edges intersect except possibly in a common endpoint.
Given two respective plane drawings of directed multigraphs $G_1$ and $G_2$, we say the drawings are \emph{disjoint} if they do not use any common points in the plane.
Moreover, the drawings are \emph{entangled} if for each cycle in either drawing, both its interior and exterior contain a point used by the respective other drawing.
Finally, for fixed vertices $v_1$ in $G_1$ and $v_2$ in $G_2$ we say the two drawings are \emph{tangent} in $v_1$ and $v_2$ if the points corresponding to $v_1$ and $v_2$ can be connected by an additional curve without intersecting any points already used in either drawing.


\begin{restatable}{lemma}{entangledtrees}
  \label{lem:entangledtrees}
  Let $G_1$ and $G_2$ be finite, directed multigraphs that are root-oriented towards $v_1$ and $v_2$, respectively.
  Then, there exist plane drawings of $G_1$ and $G_2$ that are disjoint, entangled and tangent in $v_1$ and $v_2$ if and only if both $G_1$ and $G_2$ are trees.
\end{restatable}
\begin{proof}
  If $G_1$ and $G_2$ are trees then the desired drawings clearly exist.
  To prove the other direction of the lemma, it suffices to show that both $G_1$ and $G_2$ do not contain any cycles.

  Let us thus fix plane drawings of $G_1$ and $G_2$ with the desired properties and assume that $G_1$ contains a cycle $C$.
  Without loss of generality, the root $v_2$ of $G_2$ is contained in the exterior of $C$.
  Since the drawings of $G_1$ and $G_2$ are tangent in $v_1$ and $v_2$, also the root $v_1$ of $G_1$ is contained in the exterior of $C$.
  Now, from all cycles of either $G_1$ or $G_2$ that are contained in the interior of $C$, let us select a minimal cycle $C'$.
  Minimal means that $C'$ does itself not contain any other cycles in its interior.
  Such a cycle exists since $G_1$ and $G_2$ are finite.
  We assume that $C'$ again belongs to $G_1$, the other case being analogous.
  Since the drawings of $G_1$ and $G_2$ are entangled, we get a vertex of $G_2$ in the interior of $C'$.
  Starting from this vertex we now follow directed edges in $G_2$. 
  Since the drawings of $G_1$ and $G_2$ are disjoint, we never leave the interior of $C'$, which in particular means that we never reach $v_2$.
  However, since $G_2$ only has a finite number of vertices and all except for $v_2$ have an outgoing edge, we are bound to get into a cycle eventually.
  Clearly, this new cycle of $G_2$ is still contained in the interior of $C'$, in contradiction to minimality.
\end{proof}

We make two adaptations to the abstract framework from \cref{sec:framework}.
Note that all definitions and lemmas from that section extend naturally to the following setting.

A set of units $\units$ is no longer understood as a simple subset of $\frags$.
Firstly, in this section we restrict units to be segments from the set $\segs$.
Secondly, a unit can have additional information attached to it.
As an example, $\units$ could be defined as the set of directed segments.
That is, each $\unit$ in $\units$ would correspond to an element of $\segs$, but it would also have a direction.
In particular, this means that multiple elements of $\units$ can correspond to the same geometric graph.

Moreover, a combination $\cmb$ of $\units$ is still understood as a subset of $\units$.
However, we do not allow the same geometric graph to appear twice in $\cmb$.
That is, in the above example, the elements of $\cmb$ must be pairwise distinct as segments.
It is not sufficient if only their directions differ. 

We conclude by giving some definitions and conventions that will be used in the following two subsections.
We refer to \Cref{fig:drains} for illustrations.

We assume there are unique points $\hat{\pt},\check{\pt} \in \pts$ with largest and smallest $y$-coordinates, respectively.
The horizontal line through $\check{\pt}$ is called the \emph{bottom}.

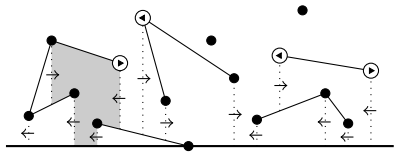
\begin{figure}[b]
  \begin{center}
    \input{fig/drains}
  \end{center}
  \caption{
    The thick line is the bottom.
    Dotted lines are borders.
    The shaded region is a face with out-degree 1.
    Points which expose a drain to the left are marked with $\protect\ptdrlft$.
    Points which expose a drain to the right are marked with $\protect\ptdrrgt$.
  }
  \label{fig:drains}
\end{figure}

For every crossing-free combination $\cmb$ of some set of units $\units$ we define a set of faces as follows.
From the endpoints of each segment $\unit$ in $\cmb$ we draw vertical rays (called \emph{borders}) downwards until we hit either the bottom or the relative interior of another segment in $\cmb$.
Then, a \emph{face} in $\cmb$ is a maximal connected region in the plane.
There is one unbounded region above the bottom, which is called the \emph{infinite face}.
The unbounded region below the bottom is not a face and will be ignored.
Furthermore, we say that two faces in $\cmb$ are \emph{adjacent} if they share a (vertical) border.

Borders are always directed either left-to-right or right-to-left.
In a combination $\cmb$, the \emph{out-degree} of a face is the number of borders directed away from that face.
We further say that a point $\pt$ \emph{exposes a drain to the left} if the border below $\pt$ is directed left-to-right and the region directly to the left of that border belongs to the infinite face.
An analogous definition is given for \emph{exposing a drain to the right}.
If a point exposes a drain either to the left or to the right, we simply say that it \emph{exposes a drain}.

\subsection{Spanning Trees}\label{subsec:st}
We define a very special set $\segsst$ of units.
Each $\unit$ in $\segsst$ is a segment from the set $\segs$ with a direction.
Additionally, below either endpoint of $\unit$ a border might be attached that is directed either left-to-right or right-to-left.

We also define the set $\cmbslst$, which contains all crossing-free combinations $\cmb$ of $\segsst$ with the following additional properties.\footnote{There is one technicality we gloss over, which however can be made precise: We also require that for each $\pt \in \cpts{\cmb}$ only the first segment $\unit$ in $\cmb$ with $\pt$ as an endpoint has a border attached below $\pt$, where first refers to the usual order induced by serializability of $\cmbscf{\segs}$. In this way, for each picture in \cref{fig:stincoherent} (a), (b) and (c), there exists a unique corresponding object in $\cmbslst$.
}
In what follows, the \emph{out-degree} of a point $\pt \in \pts$ in $\cmb$ denotes the number of segments in $\cmb$ that have $\pt$ as an endpoint and are directed away from $\pt$.

\begin{itemize}
\item The point $\hat{\pt}$ has out-degree 0 in $\cmb$.
\item Every point $\pt \in \pts$ has out-degree at most 1 in $\cmb$.
\item Every point $\pt \in \clow{\cmb}$ has out-degree 1 in $\cmb$.
\item Every finite face in $\cmb$ has out-degree 1.
\end{itemize}

Examples can be seen in \Cref{fig:stincoherent} (a), (b) and (c).
The combination in (d) violates the last three properties.

\begin{figure}
  \begin{center}
    \input{fig/stincoherent}
  \end{center}
  \caption{
    The elements of $\protect\cmbslst$ in (a) and (c) cannot be considered equivalent.
  }
  \label{fig:stincoherent}
\end{figure}

Note that all of the above properties are maintained when removing the right-most extreme element from a combination from the set $\cmbslst$.
We thus get the following lemma.

\begin{restatable}{lemma}{stserializable}
  \label{lem:st:serializable}
  For any point set $\pts$, $\cmbslst$ is serializable.
\end{restatable}

We reuse the three colors $\ptfree,\ptalive,\ptdead$ and the special marking from \cref{subsec:crossingfree} with their original meaning to describe elements of $\cmbslst$, as already depicted in \cref{fig:stincoherent}.
However, that same figure illustrates that we cannot reuse the old equivalence relation because it is not coherent.
There are three problems we have to deal with.
Firstly, the out-degree of a point can become larger than 1.
Secondly, a point with out-degree 0 can disappear in the lower shadow of a segment.
Thirdly, a finite face with out-degree not equal to 1 can be created.

To make the equivalence relation coherent, it suffices to partition the points with color $\ptalive$ into six smaller categories.
That is, we have to replace the color $\ptalive$ with 6 new colors, resulting in a total of 8 colors, and then consider two combinations equivalent if they agree in that new coloring and also in the special marking. 
We will not define the colors explicitly here, but only explain what information we have to keep track of.

For each point with color $\ptalive$ we keep track of its out-degree, that is, whether it is currently 0 or 1.
This allows us to avoid the first two problems mentioned earlier.
Furthermore, for each point with color $\ptalive$ we keep track of whether and where it exposes a drain.
This allows us to avoid the third problem because whenever a new segment $\unit$ is added to a combination $\cmb$, a new finite face is created below $\unit$, and the out-degree of that face is determined by the number of exposed drains in the lower shadow and at the endpoints of $\unit$.
Indeed, observe that borders corresponding to exposed drains in the lower shadow of $\unit$ become out-borders of the new finite face.
Also, an exposed drain at the left endpoint, say, of $\unit$ becomes an out-border of the new face if and only if it is exposed to the right.
This is why we also have to know the side a drain is exposed to.

We define an equivalence relation $\parteq$ on $\cmbslst$ based on the $2 + 2 \cdot 3 = 8$ colors from the preceding discussion, and with the usual marking.
The above intuition can be made precise, and the following can be proved.

\begin{restatable}{lemma}{stcoherent}
  \label{lem:st:coherent}
  The equivalence relation $\parteq$ on $\cmbslst$, as defined above, is coherent.
\end{restatable}

We define the set $\target \subseteq \quotient{\cmbslst}{\parteq}$ which contains $\partclass{\cmb}$ if and only if every point (except for $\hat{\pt}$) has out-degree 1 in $\cmb$ and the infinite face has out-degree 0 (equivalently, there are no exposed drains in $\cmb$).
Let $\agraphst{\pts}$ be the combination graph corresponding to the combination problem $(\segsst,\cmbslst,\parteq,\target)$.
\Cref{thm:st} with $c=8$ now follows from \cref{lem:represent} and from the following insight.

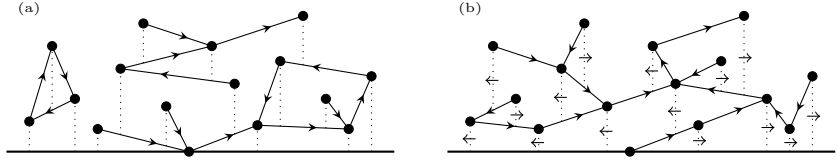
\begin{figure}
  \begin{center}
    \input{fig/trees}
  \end{center}
  \caption{
    Illustrations for the proof of \cref{lem:st:bijection}.
  }
  \label{fig:trees}
\end{figure}

\begin{lemma}
  \label{lem:st:bijection}
  There is a natural bijection between the sets $\bigcup \target$ and $\sts$ in the following sense.
  For any combination $\cmb$ in $\bigcup \target$, building the geometric graph on $\pts$ with edges that correspond to the segments in $\cmb$ yields the corresponding crossing-free spanning tree in $\sts$.
\end{lemma}
\begin{proof}
  Any $\cmb$ in $\bigcup \target$ induces two directed multigraphs $G_1$ and $G_2$ with corresponding plane drawings, as follows.
  
  $G_1$ is the graph with vertex set $\pts$ and edges that correspond to the directed segments in $\cmb$.
  By definition of $\cmbslst$ and $\target$, we at least know that $G_1$ is root-oriented towards $\hat{\pt}$, as exemplified in \cref{fig:trees} (a).

  For the vertices of $G_2$ we choose one arbitrary point in the interior of each face in $\cmb$.
  Two vertices in $G_2$ are connected if their corresponding faces in $\cmb$ are adjacent.
  The direction of that edge is chosen in accordance with the direction of the corresponding border in $\cmb$.
  Again by definition of $\cmbslst$ and $\target$, $G_2$ is root-oriented towards the vertex corresponding to the infinite face in $\cmb$, as exemplified in \cref{fig:trees} (b).
  
  It is clear that the drawings of $G_1$ and $G_2$ can be chosen such that they are disjoint, entangled, and tangent in $\hat{\pt}$ and the infinite face.
  The proof is concluded by applying \cref{lem:entangledtrees} and by observing that any spanning tree on $\pts$ can be root-oriented towards $\hat{\pt}$ in a unique way.
\end{proof}

One can prove that $c<8$ by adapting the arguments from \cref{subsec:crossingfree}.
With some more work, we get $c<7.04313$, as will be shown in \cref{sec:bounds}.

\subsection{Spanning Cycles}\label{subsec:sc}
We define a slightly different set of units $\segssc$.
In the same way as in the previous subsection, below the endpoints of any segment $\unit$ in $\segssc$ directed borders can be attached.
Here, however, the segment $\unit$ itself does not have a direction.

We also define the set $\cmbslsc$, which contains all crossing-free combinations $\cmb$ of $\segssc$ with the following additional properties.
In what follows, the \emph{degree} of a point $\pt \in \pts$ in $\cmb$ stands for the number of segments in $\cmb$ that have $\pt$ as an endpoint.
Also, if the size of $\cmb$ is $\ptsnum$, the \emph{last finite face} in $\cmb$ is defined as the face directly below the right-most extreme segment in $\cmb$.
All other finite faces are called \emph{normal}.

\begin{itemize}
\item If $|\cmb| = \ptsnum$, the last finite face in $\cmb$ has out-degree 0.
\item Every point $\pt \in \pts$ has degree at most 2 in $\cmb$.
\item Every point $\pt \in \clow{\cmb}$ has degree 2 in $\cmb$.
\item Every normal finite face in $\cmb$ has out-degree 1.
\end{itemize}

Note again that the above properties are maintained when removing the right-most extreme element from a combination from the set $\cmbslsc$.

\begin{restatable}{lemma}{scserializable}
  \label{lem:sc:serializable}
  For any point set $\pts$, $\cmbslsc$ is serializable.
\end{restatable}

It is clear that we have to keep track of the degrees of all points.
For one last time, we change the meaning of the colors $\ptfree$, $\ptalive$, $\ptdead$ and use them to identify points of degree 0, 1, and 2, respectively.

\begin{figure}
  \begin{center}
    \input{fig/scincoherent}
  \end{center}
  \caption{
    The elements of $\protect\cmbslsc$ in (a) and (c) cannot be considered equivalent.
  }
  \label{fig:scincoherent}
\end{figure}

Unsurprisingly, and similar to the previous subsection, an equivalence relation based only on these three colors and the usual marking is not coherent, as illustrated in \Cref{fig:scincoherent}.
The only problem, however, is that finite faces which do not have the right out-degree can be created.
A by now routine proof shows that these three colors are already sufficient to avoid crossings.

To avoid finite faces with wrong out-degrees, we split $\ptalive$ into 3 subcolors, and we split $\ptdead$ into 2 subcolors, giving us a total of 6 different colors which are then used to define an equivalence relation $\parteq$ on $\cmbslsc$.
For each point with color $\ptalive$ we keep track of whether it exposes a drain.
If it does, then we also keep track of whether it is to the left or to the right.
Remember, this extra information is relevant if and only if the point in question is one of the endpoints of a new segment.
For a point with color $\ptdead$ we only keep track of whether it exposes a drain or not.
The reason why this is sufficient is that such a point has, by definition, degree 2 already and cannot be an endpoint of a new segment.

\begin{restatable}{lemma}{sccoherent}
  \label{lem:sc:coherent}
  The equivalence relation $\parteq$ on $\cmbslsc$, as defined above, is coherent.
\end{restatable}

We define the set $\target \subseteq \quotient{\cmbslsc}{\parteq}$ which contains $\partclass{\cmb}$ if and only if every point has degree 2 in $\cmb$ and the infinite face has out-degree 0.
Let $\agraphsc{\pts}$ be the combination graph corresponding to $(\segssc,\cmbslsc,\parteq,\target)$.
\Cref{thm:sc} with $c=6$ now follows from \cref{lem:represent} and from the following insight.
The better bound $c < 5.61804$ will be proved in \cref{sec:bounds}.

\begin{lemma}
  \label{lem:sc:bijection}
  There is a natural bijection between the sets $\bigcup \target$ and $\scs$ in the following sense.
  For any combination $\cmb$ in $\bigcup \target$, building the geometric graph on $\pts$ with edges that correspond to the segments in $\cmb$ yields the corresponding crossing-free spanning cycle in $\scs$.
\end{lemma}
\begin{proof}
  For any $\cmb$ in $\bigcup \target$ we know that each point is of degree 2, which means that $\cmb$ is a set of disjoint cycles, as exemplified in \cref{fig:cycles} (a).
  Similar to the proof of \cref{lem:st:bijection}, $\cmb$ induces two directed multigraphs $G_E$ and $G_O$.
  
\begin{figure}[b]
  \begin{center}
    \input{fig/cycles}
  \end{center}
  \caption{
    White faces are vertices in $G_E$.
    Shaded faces are vertices in $G_O$.
  }
  \label{fig:cycles}
\end{figure}
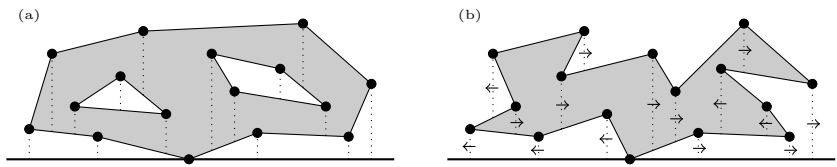
  
  The vertex set of $G_E$ is the set of faces in $\cmb$ contained in an even number of cycles.
  The vertex set of $G_O$ is the set of faces in $\cmb$ contained in an odd number of cycles.
  In both $G_E$ and $G_O$, two vertices are connected by an edge if the corresponding faces in $\cmb$ are adjacent.
  The direction of each edge reflects the direction of the corresponding border in $\cmb$.
  By definition of $\cmbslsc$ and $\target$, $G_E$ and $G_O$ are root-oriented towards the infinite face and the last finite face in $\cmb$, respectively, as exemplified in \cref{fig:cycles} (b).
  
  Clearly, there exist plane drawings of $G_E$ and $G_O$ which are disjoint, entangled, and also tangent in the respective roots.
  Applying \cref{lem:entangledtrees} hence concludes the proof.
\end{proof}

%% file: fig/drains.tex
\begin{tikzpicture}[xscale=1.5]
  \newcommand\exampleCoords{
    \coordinate (c0)  at (0.0,+0.0);
    \coordinate (c1)  at (0.2,+1.0);
    \coordinate (c2)  at (0.4,+0.3);
    \coordinate (c3)  at (0.6,-0.1);
    \coordinate (c4)  at (0.8,+0.7);
    \coordinate (c5)  at (1.0,+1.3);
    \coordinate (c6)  at (1.2,+0.2);
    \coordinate (c7)  at (1.4,-0.4);
    \coordinate (c8)  at (1.6,+1.0);
    \coordinate (c9)  at (1.8,+0.5);
    \coordinate (c10) at (2.0,-0.05);
    \coordinate (c11) at (2.2,+0.8);
    \coordinate (c12) at (2.4,+1.4);
    \coordinate (c13) at (2.6,+0.3);
    \coordinate (c14) at (2.8,-0.1);
    \coordinate (c15) at (3.0,+0.6);
  }
  \newcommand\examplePoints{
    \node[point] (p0)  at (c0) {};
    \node[point] (p1)  at (c1) {};
    \node[point] (p2)  at (c2) {};
    \node[point] (p3)  at (c3) {};
    \node[drrgt] (p4)  at (c4) {};
    \node[drlft] (p5)  at (c5) {};
    \node[point] (p6)  at (c6) {};
    \node[point] (p7)  at (c7) {};
    \node[point] (p8)  at (c8) {};
    \node[point] (p9)  at (c9) {};
    \node[point] (p10) at (c10) {};
    \node[drlft] (p11) at (c11) {};
    \node[point] (p12) at (c12) {};
    \node[point] (p13) at (c13) {};
    \node[point] (p14) at (c14) {};
    \node[drrgt] (p15) at (c15) {};
    \draw[thick] (-0.2,-0.4) -- (+3.2,-0.4);
  }
  \begin{scope}[xshift=0]
    \exampleCoords
    \fill[fill=black!20!white]
      (c1) -- (0.2,0.15) --
      (c2) -- (0.4,-0.4) --
      (0.6,-0.4) -- (c3) --
      (0.8,-0.175) -- (c4) -- cycle;
    \draw (c0) -- (c1);
    \draw (c0) -- (c2);
    \draw (c1) -- (c4);
    \draw (c3) -- (c7);
    \draw (c5) -- (c6);
    \draw (c5) -- (c9);
    \draw (c11) -- (c15);
    \draw (c10) -- (c13);
    \draw (c13) -- (c14);
    \draw[dotted,dloriented] (0.0,-0.4) -- (c0);
    \draw[dotted,droriented] (0.2,+0.15) -- (c1);
    \draw[dotted,dloriented] (0.4,-0.4) -- (c2);
    \draw[dotted,dloriented] (0.6,-0.4) -- (c3);
    \draw[dotted,dloriented] (0.8,-0.175) -- (c4);
    \draw[dotted,droriented] (1.0,-0.25) -- (c5);
    \draw[dotted,droriented] (1.2,-0.325) -- (c6);
    \draw[dotted,droriented] (1.8,-0.4) -- (c9);
    \draw[dotted,dloriented] (2.0,-0.4) -- (c10);
    \draw[dotted,droriented] (2.2,+0.0666) -- (c11);
    \draw[dotted,dloriented] (2.6,-0.4) -- (c13);
    \draw[dotted,dloriented] (2.8,-0.4) -- (c14);
    \draw[dotted,dloriented] (3.0,-0.4) -- (c15);
    \examplePoints
  \end{scope}
\end{tikzpicture}

%% file: fig/stincoherent.tex
\begin{tikzpicture}[xscale=1.5]
  \newcommand\exampleCoords{
    \coordinate (c0)  at (0.0,+0.0);
    \coordinate (c1)  at (0.2,+1.0);
    \coordinate (c2)  at (0.4,+0.3);
    \coordinate (c3)  at (0.6,-0.1);
    \coordinate (c4)  at (0.8,+0.7);
    \coordinate (c5)  at (1.0,+1.3);
    \coordinate (c6)  at (1.2,+0.2);
    \coordinate (c7)  at (1.4,-0.4);
    \coordinate (c8)  at (1.6,+1.0);
    \coordinate (c9)  at (1.8,+0.5);
    \coordinate (c10) at (2.0,-0.05);
    \coordinate (c11) at (2.2,+0.8);
    \coordinate (c12) at (2.4,+1.4);
    \coordinate (c13) at (2.6,+0.3);
    \coordinate (c14) at (2.8,-0.1);
    \coordinate (c15) at (3.0,+0.6);
  }
  \newcommand\examplePoints[1]{
    \node[alive] (p0)  at (c0) {};
    \node[alive] (p1)  at (c1) {};
    \node[dead] (p2)  at (c2) {};
    \node[dead] (p3)  at (c3) {};
    \node[alive] (p4)  at (c4) {};
    \node[alive] (p5)  at (c5) {};
    \node[dead] (p6)  at (c6) {};
    \node[dead] (p7)  at (c7) {};
    \node[alive] (p8)  at (c8) {};
    \node[alive] (p9)  at (c9) {};
    \node[alive] (p10) at (c10) {};
    \node[alive] (p11) at (c11) {};
    \node[free] (p12) at (c12) {};
    \node[dead] (p13) at (c13) {};
    \node[dead] (p14) at (c14) {};
    \node[alive] (p15) at (c15) {};
    \draw[thick] (-0.2,-0.4) -- (+3.2,-0.4);
    \node (t) at (0.0,1.5) {\tiny #1};
  }
  \begin{scope}[xshift=0]
    \exampleCoords
    \coordinate (b0) at (0.0,-0.4);
    \coordinate (b1) at (0.2,-0.0333);
    \coordinate (b2) at (0.4,-0.0666);
    \coordinate (b3) at (0.6,-0.4);
    \coordinate (b4) at (0.8,+0.3571);
    \coordinate (b5) at (1.0,+0.775);
    \coordinate (b6) at (1.2,-0.325);
    \coordinate (b7) at (1.4,-0.4);
    \coordinate (b8) at (1.6,+0.4714);
    \coordinate (b9) at (1.8,-0.1666);
    \coordinate (b10) at (2.0,-0.4);
    \coordinate (b11) at (2.2,-0.0625);
    \coordinate (b12) at (2.4,+0.75);
    \coordinate (b13) at (2.6,-0.0875);
    \coordinate (b14) at (2.8,-0.4);
    \coordinate (b15) at (3.0,-0.4);
    \draw[oriented] (c6) -- (c9);
    \draw[oriented] (c9) -- (c2);
    \draw[oriented] (c5) -- (c8);
    \draw[oriented] (c2) -- (c4);
    \draw[oriented] (c1) -- (c4);
    \draw[oriented] (c4) -- (c8);
    \draw[oriented] (c11) -- (c15);
    \draw[oriented] (c7) -- (c10);
    \draw[oriented] (c3) -- (c7);
    \draw[oriented] (c0) -- (c3);
    \draw[oriented] (c10) -- (c14);
    \draw[oriented] (c13) -- (c11);
    \draw[oriented] (c14) -- (c15);
    \draw[dotted,dloriented] (b0) -- (c0);
    \draw[dotted,dloriented] (b1) -- (c1);
    \draw[dotted,dloriented] (b2) -- (c2);
    \draw[dotted,dloriented] (b3) -- (c3);
    \draw[dotted,droriented] (b4) -- (c4);
    \draw[dotted,dloriented] (b5) -- (c5);
    \draw[dotted,dloriented] (b6) -- (c6);
    \draw[dotted,droriented] (b8) -- (c8);
    \draw[dotted,dloriented] (b9) -- (c9);
    \draw[dotted,droriented] (b10) -- (c10);
    \draw[dotted,dloriented] (b11) -- (c11);
    \draw[dotted,dloriented] (b13) -- (c13);
    \draw[dotted,droriented] (b14) -- (c14);
    \draw[dotted,droriented] (b15) -- (c15);
    \examplePoints{(a)}
    \node[marked] at (c11) {};
    \node at (3.6,0.4) {$\ctarrow{}{}{\unit}$};
  \end{scope}
  \begin{scope}[xshift=120]
    \exampleCoords
    \coordinate (b0) at (0.0,-0.4);
    \coordinate (b1) at (0.2,-0.0333);
    \coordinate (b2) at (0.4,-0.0666);
    \coordinate (b3) at (0.6,-0.4);
    \coordinate (b4) at (0.8,+0.3571);
    \coordinate (b5) at (1.0,+0.775);
    \coordinate (b6) at (1.2,-0.325);
    \coordinate (b7) at (1.4,-0.4);
    \coordinate (b8) at (1.6,+0.4714);
    \coordinate (b9) at (1.8,-0.1666);
    \coordinate (b10) at (2.0,-0.4);
    \coordinate (b11) at (2.2,-0.0625);
    \coordinate (b12) at (2.4,+0.75);
    \coordinate (b13) at (2.6,-0.0875);
    \coordinate (b14) at (2.8,-0.4);
    \coordinate (b15) at (3.0,-0.4);
    \fill[fill=black!20!white]
      (c8) -- (b8) -- (c9) -- (b9) --
      (c10) -- (b11) -- (c11) --
      (b12) -- (c12) -- cycle; 
    \draw[oriented] (c6) -- (c9);
    \draw[oriented] (c9) -- (c2);
    \draw[oriented] (c5) -- (c8);
    \draw[oriented] (c2) -- (c4);
    \draw[oriented] (c1) -- (c4);
    \draw[oriented] (c4) -- (c8);
    \draw[oriented] (c11) -- (c15);
    \draw[oriented] (c7) -- (c10);
    \draw[oriented] (c3) -- (c7);
    \draw[oriented] (c0) -- (c3);
    \draw[oriented] (c10) -- (c14);
    \draw[oriented] (c13) -- (c11);
    \draw[oriented] (c14) -- (c15);
    \draw[oriented] (c8) -- node[above,pos=0.4] {$\scriptstyle\unit$} (c12);
    \draw[dotted,dloriented] (b0) -- (c0);
    \draw[dotted,dloriented] (b1) -- (c1);
    \draw[dotted,dloriented] (b2) -- (c2);
    \draw[dotted,dloriented] (b3) -- (c3);
    \draw[dotted,droriented] (b4) -- (c4);
    \draw[dotted,dloriented] (b5) -- (c5);
    \draw[dotted,dloriented] (b6) -- (c6);
    \draw[dotted,droriented] (b8) -- (c8);
    \draw[dotted,dloriented] (b9) -- (c9);
    \draw[dotted,droriented] (b10) -- (c10);
    \draw[dotted,dloriented] (b11) -- (c11);
    \draw[dotted,dloriented] (b13) -- (c13);
    \draw[dotted,droriented] (b14) -- (c14);
    \draw[dotted,droriented] (b15) -- (c15);
    \draw[dotted,dloriented] (b12) -- (c12);
    \examplePoints{(b)}
    \node[dead] at (c9) {};
    \node[dead] at (c10) {};
    \node[dead] at (c11) {};
    \node[alive] at (c12) {};
    \node[marked] at (c8) {};
  \end{scope}

  \begin{scope}[xshift=0,yshift=-70]
    \exampleCoords
    \coordinate (b0) at (0.0,-0.4);
    \coordinate (b1) at (0.2,+0.15);
    \coordinate (b2) at (0.4,-0.4);
    \coordinate (b3) at (0.6,-0.4);
    \coordinate (b4) at (0.8,+0.25);
    \coordinate (b5) at (1.0,+0.225);
    \coordinate (b6) at (1.2,-0.325);
    \coordinate (b7) at (1.4,-0.4);
    \coordinate (b8) at (1.6,+0.4);
    \coordinate (b9) at (1.8,-0.1666);
    \coordinate (b10) at (2.0,-0.4);
    \coordinate (b11) at (2.2,+0.0666);
    \coordinate (b12) at (2.4,+0.75);
    \coordinate (b13) at (2.6,-0.0875);
    \coordinate (b14) at (2.8,-0.4);
    \coordinate (b15) at (3.0,-0.4);
    \draw[oriented] (c7) -- (c10);
    \draw[oriented] (c2) -- (c6);
    \draw[oriented] (c8) -- (c5);
    \draw[oriented] (c1) -- (c4);
    \draw[oriented] (c4) -- (c5);
    \draw[oriented] (c15) -- (c11);
    \draw[oriented] (c0) -- (c2);
    \draw[oriented] (c6) -- (c9);
    \draw[oriented] (c3) -- (c7);
    \draw[oriented] (c13) -- (c10);
    \draw[oriented] (c10) -- (c9);
    \draw[oriented] (c14) -- (c10);
    \draw[dotted,dloriented] (b0) -- (c0);
    \draw[dotted,dloriented] (b1) -- (c1);
    \draw[dotted,dloriented] (b2) -- (c2);
    \draw[dotted,dloriented] (b3) -- (c3);
    \draw[dotted,dloriented] (b4) -- (c4);
    \draw[dotted,dloriented] (b5) -- (c5);
    \draw[dotted,dloriented] (b6) -- (c6);
    \draw[dotted,dloriented] (b8) -- (c8);
    \draw[dotted,dloriented] (b9) -- (c9);
    \draw[dotted,droriented] (b10) -- (c10);
    \draw[dotted,droriented] (b11) -- (c11);
    \draw[dotted,droriented] (b13) -- (c13);
    \draw[dotted,droriented] (b14) -- (c14);
    \draw[dotted,droriented] (b15) -- (c15);
    \examplePoints{(c)}
    \node[marked] at (c11) {};
    \node at (3.6,0.4) {$\ctnarrow{}{}{\unit}$};
  \end{scope}

  \begin{scope}[xshift=120,yshift=-70]
    \exampleCoords
    \coordinate (b0) at (0.0,-0.4);
    \coordinate (b1) at (0.2,+0.15);
    \coordinate (b2) at (0.4,-0.4);
    \coordinate (b3) at (0.6,-0.4);
    \coordinate (b4) at (0.8,+0.25);
    \coordinate (b5) at (1.0,+0.225);
    \coordinate (b6) at (1.2,-0.325);
    \coordinate (b7) at (1.4,-0.4);
    \coordinate (b8) at (1.6,+0.4);
    \coordinate (b9) at (1.8,-0.1666);
    \coordinate (b10) at (2.0,-0.4);
    \coordinate (b11) at (2.2,+0.0666);
    \coordinate (b12) at (2.4,+0.75);
    \coordinate (b13) at (2.6,-0.0875);
    \coordinate (b14) at (2.8,-0.4);
    \coordinate (b15) at (3.0,-0.4);
    \fill[fill=black!20!white]
      (c8) -- (b8) -- (c9) --
      (c10) -- (b11) -- (c11) -- (b12) --
      (c12) -- cycle;
    \draw[oriented] (c7) -- (c10);
    \draw[oriented] (c2) -- (c6);
    \draw[oriented] (c8) -- (c5);
    \draw[oriented] (c1) -- (c4);
    \draw[oriented] (c4) -- (c5);
    \draw[oriented] (c15) -- (c11);
    \draw[oriented] (c0) -- (c2);
    \draw[oriented] (c6) -- (c9);
    \draw[oriented] (c3) -- (c7);
    \draw[oriented] (c13) -- (c10);
    \draw[oriented] (c10) -- (c9);
    \draw[oriented] (c14) -- (c10);
    \draw[oriented] (c8) -- node[above,pos=0.4] {$\scriptstyle\unit$} (c12);
    \draw[dotted,dloriented] (b0) -- (c0);
    \draw[dotted,dloriented] (b1) -- (c1);
    \draw[dotted,dloriented] (b2) -- (c2);
    \draw[dotted,dloriented] (b3) -- (c3);
    \draw[dotted,dloriented] (b4) -- (c4);
    \draw[dotted,dloriented] (b5) -- (c5);
    \draw[dotted,dloriented] (b6) -- (c6);
    \draw[dotted,dloriented] (b8) -- (c8);
    \draw[dotted,dloriented] (b9) -- (c9);
    \draw[dotted,droriented] (b10) -- (c10);
    \draw[dotted,droriented] (b11) -- (c11);
    \draw[dotted,dloriented] (b12) -- (c12);
    \draw[dotted,droriented] (b13) -- (c13);
    \draw[dotted,droriented] (b14) -- (c14);
    \draw[dotted,droriented] (b15) -- (c15);
    \examplePoints{(d)}
    \node[dead] at (c9) {};
    \node[dead] at (c10) {};
    \node[dead] at (c11) {};
    \node[alive] at (c12) {};
    \node[marked] at (c8) {};
  \end{scope}

\end{tikzpicture}

%% file: fig/trees.tex
          \begin{tikzpicture}[xscale=1.5]
            \newcommand\exampleCoords{
              \coordinate (c0)  at (0.0,+0.0);
              \coordinate (c1)  at (0.2,+1.0);
              \coordinate (c2)  at (0.4,+0.3);
              \coordinate (c3)  at (0.6,-0.1);
              \coordinate (c4)  at (0.8,+0.7);
              \coordinate (c5)  at (1.0,+1.3);
              \coordinate (c6)  at (1.2,+0.2);
              \coordinate (c7)  at (1.4,-0.4);
              \coordinate (c8)  at (1.6,+1.0);
              \coordinate (c9)  at (1.8,+0.5);
              \coordinate (c10) at (2.0,-0.05);
              \coordinate (c11) at (2.2,+0.8);
              \coordinate (c12) at (2.4,+1.4);
              \coordinate (c13) at (2.6,+0.3);
              \coordinate (c14) at (2.8,-0.1);
              \coordinate (c15) at (3.0,+0.6);
            }
            \newcommand\examplePoints[1]{
              \node[point] (p0)  at (c0) {};
              \node[point] (p1)  at (c1) {};
              \node[point] (p2)  at (c2) {};
              \node[point] (p3)  at (c3) {};
              \node[point] (p4)  at (c4) {};
              \node[point] (p5)  at (c5) {};
              \node[point] (p6)  at (c6) {};
              \node[point] (p7)  at (c7) {};
              \node[point] (p8)  at (c8) {};
              \node[point] (p9)  at (c9) {};
              \node[point] (p10) at (c10) {};
              \node[point] (p11) at (c11) {};
              \node[point] (p12) at (c12) {};
              \node[point] (p13) at (c13) {};
              \node[point] (p14) at (c14) {};
              \node[point] (p15) at (c15) {};
              \draw[thick] (-0.2,-0.4) -- (+3.2,-0.4);
              \node (t) at (0.0,1.5) {\tiny #1};
            }
            \begin{scope}[xshift=-110]
              \exampleCoords
              \coordinate (b0) at (0.0,-0.4);
              \coordinate (b1) at (0.2,+0.15);
              \coordinate (b2) at (0.4,-0.4);
              \coordinate (b3) at (0.6,-0.4);
              \coordinate (b4) at (0.8,-0.175);
              \coordinate (b5) at (1.0,+0.775);
              \coordinate (b6) at (1.2,-0.325);
              \coordinate (b7) at (1.4,-0.4);
              \coordinate (b8) at (1.6,+0.54);
              \coordinate (b9) at (1.8,-0.1666);
              \coordinate (b10) at (2.0,-0.4);
              \coordinate (b11) at (2.2,-0.0625);
              \coordinate (b12) at (2.4,+0.75);
              \coordinate (b13) at (2.6,-0.0875);
              \coordinate (b14) at (2.8,-0.4);
              \coordinate (b15) at (3.0,-0.4);
              
              \draw[oriented] (c0) -- (c1); 
              \draw[oriented] (c1) -- (c2); 
              \draw[oriented] (c2) -- (c0); 
              \draw[oriented] (c4) -- (c8); 
              \draw[oriented] (c5) -- (c8); 
              \draw[oriented] (c8) -- (c12); 
              \draw[oriented] (c3) -- (c7); 
              \draw[oriented] (c6) -- (c7); 
              \draw[oriented] (c7) -- (c10); 
              \draw[oriented] (c9) -- (c4); 
              \draw[oriented] (c10) -- (c14); 
              \draw[oriented] (c14) -- (c15); 
              \draw[oriented] (c15) -- (c11); 
              \draw[oriented] (c11) -- (c10); 
              \draw[oriented] (c13) -- (c14); 
              
              \draw[dotted] (b0) -- (c0);
              \draw[dotted] (b1) -- (c1);
              \draw[dotted] (b2) -- (c2);
              \draw[dotted] (b3) -- (c3);
              \draw[dotted] (b4) -- (c4);
              \draw[dotted] (b5) -- (c5);
              \draw[dotted] (b6) -- (c6);
              \draw[dotted] (b8) -- (c8);
              \draw[dotted] (b9) -- (c9);
              \draw[dotted] (b10) -- (c10);
              \draw[dotted] (b11) -- (c11);
              \draw[dotted] (b12) -- (c12);
              \draw[dotted] (b13) -- (c13);
              \draw[dotted] (b14) -- (c14);
              \draw[dotted] (b15) -- (c15);
              \examplePoints{(a)}
            \end{scope}
            \begin{scope}[xshift=0]
              \exampleCoords
 
              \coordinate (b0) at (0.0,-0.4);
              \coordinate (b1) at (0.2,+0.15);
              \coordinate (b2) at (0.4,-0.0666);
              \coordinate (b3) at (0.6,-0.4);
              \coordinate (b4) at (0.8,-0.0);
              \coordinate (b5) at (1.0,+0.45);
              \coordinate (b6) at (1.2,-0.4);
              \coordinate (b7) at (1.4,-0.4);
              \coordinate (b8) at (1.6,+0.4);
              \coordinate (b9) at (1.8,-0.1666);
              \coordinate (b10) at (2.0,-0.4);
              \coordinate (b11) at (2.2,+0.4);
              \coordinate (b12) at (2.4,+0.35);
              \coordinate (b13) at (2.6,-0.4);
              \coordinate (b14) at (2.8,-0.4);
              \coordinate (b15) at (3.0,-0.4);

              \draw[oriented] (c2) -- (c0);
              \draw[oriented] (c0) -- (c3);
              \draw[oriented] (c7) -- (c10);
              \draw[oriented] (c3) -- (c6);
              \draw[oriented] (c1) -- (c4);
              \draw[oriented] (c5) -- (c4);
              \draw[oriented] (c4) -- (c6);
              \draw[oriented] (c6) -- (c9);
              \draw[oriented] (c9) -- (c8);
              \draw[oriented] (c8) -- (c12);
              \draw[oriented] (c11) -- (c9);
              \draw[oriented] (c13) -- (c9);
              \draw[oriented] (c10) -- (c13);
              \draw[oriented] (c14) -- (c13);
              \draw[oriented] (c15) -- (c14);

              \draw[dotted,dloriented] (b0) -- (c0);
              \draw[dotted,dloriented] (b1) -- (c1);
              \draw[dotted,droriented] (b2) -- (c2);
              \draw[dotted,dloriented] (b3) -- (c3);
              \draw[dotted,dloriented] (b4) -- (c4);
              \draw[dotted,droriented] (b5) -- (c5);
              \draw[dotted,dloriented] (b6) -- (c6);
              \draw[dotted,dloriented] (b8) -- (c8);
              \draw[dotted,dloriented] (b9) -- (c9);
              \draw[dotted,droriented] (b10) -- (c10);
              \draw[dotted,droriented] (b11) -- (c11);
              \draw[dotted,droriented] (b12) -- (c12);
              \draw[dotted,droriented] (b13) -- (c13);
              \draw[dotted,droriented] (b14) -- (c14);
              \draw[dotted,droriented] (b15) -- (c15);
              \examplePoints{(b)}
            \end{scope}
          \end{tikzpicture}

%% file: fig/scincoherent.tex
\begin{tikzpicture}[xscale=1.5]
  \newcommand\exampleCoords{
    \coordinate (c0)  at (0.0,+0.0);
    \coordinate (c1)  at (0.2,+1.0);
    \coordinate (c2)  at (0.4,+0.3);
    \coordinate (c3)  at (0.6,-0.1);
    \coordinate (c4)  at (0.8,+0.7);
    \coordinate (c5)  at (1.0,+1.3);
    \coordinate (c6)  at (1.2,+0.2);
    \coordinate (c7)  at (1.4,-0.4);
    \coordinate (c8)  at (1.6,+1.0);
    \coordinate (c9)  at (1.8,+0.5);
    \coordinate (c10) at (2.0,-0.05);
    \coordinate (c11) at (2.2,+0.8);
    \coordinate (c12) at (2.4,+1.4);
    \coordinate (c13) at (2.6,+0.3);
    \coordinate (c14) at (2.8,-0.1);
    \coordinate (c15) at (3.0,+0.6);
  }
  \newcommand\examplePoints[1]{
    \node[dead] (p0)  at (c0) {};
    \node[alive] (p1)  at (c1) {};
    \node[alive] (p2)  at (c2) {};
    \node[dead] (p3)  at (c3) {};
    \node[alive] (p4)  at (c4) {};
    \node[free] (p5)  at (c5) {};
    \node[dead] (p6)  at (c6) {};
    \node[dead] (p7)  at (c7) {};
    \node[free] (p8)  at (c8) {};
    \node[dead] (p9)  at (c9) {};
    \node[dead] (p10) at (c10) {};
    \node[free] (p11) at (c11) {};
    \node[free] (p12) at (c12) {};
    \node[free] (p13) at (c13) {};
    \node[alive] (p14) at (c14) {};
    \node[free] (p15) at (c15) {};
    \draw[thick] (-0.2,-0.4) -- (+3.2,-0.4);
    \node (t) at (0.0,1.5) {\tiny #1};
  }
  \begin{scope}[xshift=0]
    \exampleCoords
    \coordinate (b0) at (0.0,-0.4);
    \coordinate (b1) at (0.2,-0.0333);
    \coordinate (b2) at (0.4,-0.0666);
    \coordinate (b3) at (0.6,-0.4);
    \coordinate (b4) at (0.8,+0.3857);
    \coordinate (b5) at (1.0,-0.4);
    \coordinate (b6) at (1.2,-0.325);
    \coordinate (b7) at (1.4,-0.4);
    \coordinate (b8) at (1.6,+0.54);
    \coordinate (b9) at (1.8,+0.0875);
    \coordinate (b10) at (2.0,-0.4);
    \coordinate (b11) at (2.2,+0.0125);
    \coordinate (b12) at (2.4,-0.4);
    \coordinate (b13) at (2.6,-0.4);
    \coordinate (b14) at (2.8,-0.4);
    \coordinate (b15) at (3.0,-0.4);
    \draw (c14) -- (c6);
    \draw (c6) -- (c10);
    \draw (c10) -- (c7);
    \draw (c7) -- (c3);
    \draw (c3) -- (c0);
    \draw (c0) -- (c1);
    \draw (c2) -- (c9);
    \draw (c9) -- (c4);
    \draw[dotted,dloriented] (b0) -- (c0);
    \draw[dotted,droriented] (b1) -- (c1);
    \draw[dotted,dloriented] (b2) -- (c2);
    \draw[dotted,dloriented] (b3) -- (c3);
    \draw[dotted,dloriented] (b4) -- (c4);
    \draw[dotted,dloriented] (b6) -- (c6);
    \draw[dotted,dloriented] (b9) -- (c9);
    \draw[dotted,droriented] (b10) -- (c10);
    \draw[dotted,droriented] (b14) -- (c14);
    \examplePoints{(a)}
    \node[marked] at (c4) {};
    \node at (3.6,0.4) {$\ctarrow{}{}{\unit}$};
  \end{scope}
  \begin{scope}[xshift=120]
    \exampleCoords
    \coordinate (b0) at (0.0,-0.4);
    \coordinate (b1) at (0.2,-0.0333);
    \coordinate (b2) at (0.4,-0.0666);
    \coordinate (b3) at (0.6,-0.4);
    \coordinate (b4) at (0.8,+0.3857);
    \coordinate (b5) at (1.0,-0.4);
    \coordinate (b6) at (1.2,-0.325);
    \coordinate (b7) at (1.4,-0.4);
    \coordinate (b8) at (1.6,+0.54);
    \coordinate (b9) at (1.8,+0.0875);
    \coordinate (b10) at (2.0,-0.4);
    \coordinate (b11) at (2.2,+0.0125);
    \coordinate (b12) at (2.4,-0.4);
    \coordinate (b13) at (2.6,-0.4);
    \coordinate (b14) at (2.8,-0.4);
    \coordinate (b15) at (3.0,-0.4);
    \fill[black!20!white] (c8) -- (b8) -- (c9) -- (b9) -- (b11) -- (c11) -- cycle;
    \draw (c14) -- (c6);
    \draw (c6) -- (c10);
    \draw (c10) -- (c7);
    \draw (c7) -- (c3);
    \draw (c3) -- (c0);
    \draw (c0) -- (c1);
    \draw (c2) -- (c9);
    \draw (c9) -- (c4);
    \draw (c8) -- node[above,pos=0.55] {$\scriptstyle\unit$} (c11);
    \draw[dotted,dloriented] (b0) -- (c0);
    \draw[dotted,droriented] (b1) -- (c1);
    \draw[dotted,dloriented] (b2) -- (c2);
    \draw[dotted,dloriented] (b3) -- (c3);
    \draw[dotted,dloriented] (b4) -- (c4);
    \draw[dotted,dloriented] (b6) -- (c6);
    \draw[dotted,dloriented] (b9) -- (c9);
    \draw[dotted,droriented] (b10) -- (c10);
    \draw[dotted,droriented] (b14) -- (c14);
    \draw[dotted,droriented] (b8) -- (c8);
    \draw[dotted,dloriented] (b11) -- (c11);
    \examplePoints{(b)}
    \node[alive] at (c8) {};
    \node[alive] at (c11) {};
    \node[marked] at (c8) {};
  \end{scope}

  \begin{scope}[xshift=0,yshift=-70]
    \exampleCoords
    \coordinate (b0) at (0.0,-0.4);
    \coordinate (b1) at (0.2,-0.0333);
    \coordinate (b2) at (0.4,-0.0666);
    \coordinate (b3) at (0.6,-0.4);
    \coordinate (b4) at (0.8,+0.25);
    \coordinate (b5) at (1.0,-0.4);
    \coordinate (b6) at (1.2,-0.0785);
    \coordinate (b7) at (1.4,-0.4);
    \coordinate (b8) at (1.6,+0.54);
    \coordinate (b9) at (1.8,-0.0571);
    \coordinate (b10) at (2.0,-0.2714);
    \coordinate (b11) at (2.2,-0.228);
    \coordinate (b12) at (2.4,-0.4);
    \coordinate (b13) at (2.6,-0.4);
    \coordinate (b14) at (2.8,-0.4);
    \coordinate (b15) at (3.0,-0.4);
    \draw (c14) -- (c7);
    \draw (c7) -- (c10);
    \draw (c10) -- (c3);
    \draw (c3) -- (c0);
    \draw (c0) -- (c1);
    \draw (c2) -- (c6);
    \draw (c6) -- (c9);
    \draw (c9) -- (c4);
    \draw[dotted,dloriented] (b0) -- (c0);
    \draw[dotted,droriented] (b1) -- (c1);
    \draw[dotted,droriented] (b2) -- (c2);
    \draw[dotted,dloriented] (b3) -- (c3);
    \draw[dotted,dloriented] (b4) -- (c4);
    \draw[dotted,droriented] (b6) -- (c6);
    \draw[dotted,droriented] (b9) -- (c9);
    \draw[dotted,droriented] (b10) -- (c10);
    \draw[dotted,droriented] (b14) -- (c14);
    \examplePoints{(c)}
    \node[marked] at (c4) {};
    \node at (3.6,0.4) {$\ctnarrow{}{}{}$};
  \end{scope}

  \begin{scope}[xshift=120,yshift=-70]
    \exampleCoords
    \coordinate (b0) at (0.0,-0.4);
    \coordinate (b1) at (0.2,-0.0333);
    \coordinate (b2) at (0.4,-0.0666);
    \coordinate (b3) at (0.6,-0.4);
    \coordinate (b4) at (0.8,+0.25);
    \coordinate (b5) at (1.0,-0.4);
    \coordinate (b6) at (1.2,-0.0785);
    \coordinate (b7) at (1.4,-0.4);
    \coordinate (b8) at (1.6,+0.54);
    \coordinate (b9) at (1.8,-0.0571);
    \coordinate (b10) at (2.0,-0.2714);
    \coordinate (b11) at (2.2,-0.228);
    \coordinate (b12) at (2.4,-0.4);
    \coordinate (b13) at (2.6,-0.4);
    \coordinate (b14) at (2.8,-0.4);
    \coordinate (b15) at (3.0,-0.4);
    \fill[black!20!white] (c8) -- (b8) -- (c9) -- (b9) --
      (c10) -- (b10) -- (b11) -- (c11) -- cycle; 
    \draw (c14) -- (c7);
    \draw (c7) -- (c10);
    \draw (c10) -- (c3);
    \draw (c3) -- (c0);
    \draw (c0) -- (c1);
    \draw (c2) -- (c6);
    \draw (c6) -- (c9);
    \draw (c9) -- (c4);
    \draw (c8) -- node[above,pos=0.55] {$\scriptstyle\unit$} (c11);
    \draw[dotted,dloriented] (b0) -- (c0);
    \draw[dotted,droriented] (b1) -- (c1);
    \draw[dotted,droriented] (b2) -- (c2);
    \draw[dotted,dloriented] (b3) -- (c3);
    \draw[dotted,dloriented] (b4) -- (c4);
    \draw[dotted,droriented] (b6) -- (c6);
    \draw[dotted,droriented] (b9) -- (c9);
    \draw[dotted,droriented] (b10) -- (c10);
    \draw[dotted,droriented] (b14) -- (c14);
    \draw[dotted,droriented] (b8) -- (c8);
    \draw[dotted,dloriented] (b11) -- (c11);
    \examplePoints{(d)}
    \node[alive] at (c8) {};
    \node[alive] at (c11) {};
    \node[marked] at (c8) {};
  \end{scope}

\end{tikzpicture}

%% file: fig/cycles.tex
          \begin{tikzpicture}[xscale=1.5]
            \newcommand\exampleCoords{
              \coordinate (c0)  at (0.0,+0.0);
              \coordinate (c1)  at (0.2,+1.0);
              \coordinate (c2)  at (0.4,+0.3);
              \coordinate (c3)  at (0.6,-0.1);
              \coordinate (c4)  at (0.8,+0.7);
              \coordinate (c5)  at (1.0,+1.3);
              \coordinate (c6)  at (1.2,+0.2);
              \coordinate (c7)  at (1.4,-0.4);
              \coordinate (c8)  at (1.6,+1.0);
              \coordinate (c9)  at (1.8,+0.5);
              \coordinate (c10) at (2.0,-0.05);
              \coordinate (c11) at (2.2,+0.8);
              \coordinate (c12) at (2.4,+1.4);
              \coordinate (c13) at (2.6,+0.3);
              \coordinate (c14) at (2.8,-0.1);
              \coordinate (c15) at (3.0,+0.6);
            }
            \newcommand\examplePoints[1]{
              \node[point] (p0)  at (c0) {};
              \node[point] (p1)  at (c1) {};
              \node[point] (p2)  at (c2) {};
              \node[point] (p3)  at (c3) {};
              \node[point] (p4)  at (c4) {};
              \node[point] (p5)  at (c5) {};
              \node[point] (p6)  at (c6) {};
              \node[point] (p7)  at (c7) {};
              \node[point] (p8)  at (c8) {};
              \node[point] (p9)  at (c9) {};
              \node[point] (p10) at (c10) {};
              \node[point] (p11) at (c11) {};
              \node[point] (p12) at (c12) {};
              \node[point] (p13) at (c13) {};
              \node[point] (p14) at (c14) {};
              \node[point] (p15) at (c15) {};
              \draw[thick] (-0.2,-0.4) -- (+3.2,-0.4);
              \node (t) at (0.0,1.5) {\tiny #1};
            }
            \begin{scope}[xshift=110]
              \exampleCoords
              \coordinate (b0) at (0.0,-0.4);
              \coordinate (b1) at (0.2,-0.0333);
              \coordinate (b2) at (0.4,-0.0666);
              \coordinate (b3) at (0.6,-0.4);
              \coordinate (b4) at (0.8,+0.25);
              \coordinate (b5) at (1.0,+0.45);
              \coordinate (b6) at (1.2,-0.325);
              \coordinate (b7) at (1.4,-0.4);
              \coordinate (b8) at (1.6,-0.28333);
              \coordinate (b9) at (1.8,-0.1666);
              \coordinate (b10) at (2.0,-0.4);
              \coordinate (b11) at (2.2,+0.4);
              \coordinate (b12) at (2.4,+0.55);
              \coordinate (b13) at (2.6,-0.0875);
              \coordinate (b14) at (2.8,-0.4);
              \coordinate (b15) at (3.0,-0.4);
              
              \path[fill=black!20!white] 
                (c0) -- (c1) -- (c5) -- 
                (c12) -- (c15) -- (c14) -- 
                (c10) -- (c7) -- (c3) -- cycle;
              \path[fill=white] 
                (c2) -- (c4) -- (c6) -- cycle;
              \path[fill=white]
                (c8) -- (c9) -- (c13) --
                (c11) -- cycle;

              \draw (c0) -- (c1);
              \draw (c0) -- (c3);
              \draw (c1) -- (c5);
              \draw (c3) -- (c7);
              \draw (c5) -- (c12);
              \draw (c7) -- (c10);
              \draw (c10) -- (c14);
              \draw (c12) -- (c15);
              \draw (c14) -- (c15);
              \draw (c2) -- (c4);
              \draw (c2) -- (c6);
              \draw (c4) -- (c6);
              \draw (c8) -- (c9);
              \draw (c8) -- (c11);
              \draw (c9) -- (c13);
              \draw (c11) -- (c13);

              \draw[dotted] (b0) -- (c0);
              \draw[dotted] (b1) -- (c1);
              \draw[dotted] (b2) -- (c2);
              \draw[dotted] (b3) -- (c3);
              \draw[dotted] (b4) -- (c4);
              \draw[dotted] (b5) -- (c5);
              \draw[dotted] (b6) -- (c6);
              \draw[dotted] (b7) -- (c7);
              \draw[dotted] (b8) -- (c8);
              \draw[dotted] (b9) -- (c9);
              \draw[dotted] (b10) -- (c10);
              \draw[dotted] (b11) -- (c11);
              \draw[dotted] (b12) -- (c12);
              \draw[dotted] (b13) -- (c13);
              \draw[dotted] (b14) -- (c14);
              \draw[dotted] (b15) -- (c15);
              \examplePoints{(a)}
            \end{scope}
            \begin{scope}[xshift=220]
              \exampleCoords
              \coordinate (b0) at (0.0,-0.4);
              \coordinate (b1) at (0.2,+0.15);
              \coordinate (b2) at (0.4,-0.0666);
              \coordinate (b3) at (0.6,-0.4);
              \coordinate (b4) at (0.8,-0.0);
              \coordinate (b5) at (1.0,+0.775);
              \coordinate (b6) at (1.2,-0.4);
              \coordinate (b7) at (1.4,-0.4);
              \coordinate (b8) at (1.6,-0.28333);
              \coordinate (b9) at (1.8,-0.1666);
              \coordinate (b10) at (2.0,-0.4);
              \coordinate (b11) at (2.2,-0.0675);
              \coordinate (b12) at (2.4,+0.75);
              \coordinate (b13) at (2.6,-0.0875);
              \coordinate (b14) at (2.8,-0.4);
              \coordinate (b15) at (3.0,-0.4);

              \path[fill=black!20!white]
                (c0) -- (c2) -- (c1) -- 
                (c5) -- (c4) -- (c8) --
                (c9) -- (c12) -- (c15) --
                (c11) -- (c13) -- (c14) --
                (c10) -- (c7) -- (c6) --
                (c3) -- cycle;

              \draw (c0) -- (c2);
              \draw (c0) -- (c3);
              \draw (c1) -- (c2);
              \draw (c1) -- (c5);
              \draw (c3) -- (c6);
              \draw (c4) -- (c5);
              \draw (c4) -- (c8);
              \draw (c6) -- (c7);
              \draw (c7) -- (c10);
              \draw (c8) -- (c9);
              \draw (c9) -- (c12);
              \draw (c10) -- (c14);
              \draw (c12) -- (c15);
              \draw (c11) -- (c13);
              \draw (c13) -- (c14);
              \draw (c11) -- (c15);

              \draw[dotted,dloriented] (b0) -- (c0);
              \draw[dotted,dloriented] (b1) -- (c1);
              \draw[dotted,droriented] (b2) -- (c2);
              \draw[dotted,dloriented] (b3) -- (c3);
              \draw[dotted,droriented] (b4) -- (c4);
              \draw[dotted,droriented] (b5) -- (c5);
              \draw[dotted,dloriented] (b6) -- (c6);
              \draw[dotted,droriented] (b8) -- (c8);
              \draw[dotted,droriented] (b9) -- (c9);
              \draw[dotted,droriented] (b10) -- (c10);
              \draw[dotted,dloriented] (b11) -- (c11);
              \draw[dotted,droriented] (b12) -- (c12);
              \draw[dotted,dloriented] (b13) -- (c13);
              \draw[dotted,droriented] (b14) -- (c14);
              \draw[dotted,droriented] (b15) -- (c15);
              \examplePoints{(b)}
            \end{scope}
          \end{tikzpicture}

%% file: bounds.tex
\section{Bounding the Size of Combination Graphs}\label{sec:bounds}

In this final section, we show how to prove the bounds on the constants $c$ in \cref{thm:pg,thm:st,thm:sc}.
Recall that vertices of combination graphs can be interpreted as colorings of $\pts$ with a finite number of colors.
The following proofs make use of the fact that certain patterns cannot occur in these colorings because of geometric constraints.
Additionally, in the case of spanning trees and spanning cycles, we can improve the bounds further by identifying and discarding vertices from which the sink $\cgsink$ can no longer be reached.

\subsection{All geometric graphs}

Recall the definition of the set $\cmbscf{\segs}$ and the corresponding equivalence relation $\parteq$ from \cref{subsec:crossingfree}.
We are left to prove the following lemma.

{\renewcommand\footnote[1]{}\cfbound*}
\begin{proof}
  We may encode equivalence classes $\partclass{\cmb}$ in $\quotient{\cmbscf{\segs}}{\parteq}$ by a string $s_\cmb$ of length $\ptsnum$ over the alphabet $\{\ptdead,\ptalive,\ptfree\}$ and by an index $k_\cmb$.
  The $i$-th entry in $s_\cmb$ is $\ptdead$ if $\pt_i \in \clow{\cmb}$, it is $\ptalive$ if $\pt_i \in \cpts{\cmb} \setminus \clow{\cmb}$, and it is $\ptfree$ otherwise.
  The index $k_\cmb$ is equal to the number $i$ which satisfies $\pt_i = \ulft{\umaxin{\cmb}}$.
  From this encoding we immediately get a bound of $O(3^nn)$ on the size of $\quotient{\cmbscf{\segs}}{\parteq}$.

  Our proof strategy now is as follows.
  We ignore the indices $k_\cmb$ and give an upper bound of $O(\alpha^n)$ on the number of strings $s_\cmb$ which correspond to a combination $\cmb$ in $\cmbscf{\segs}$.
  From this, the desired upper bound easily follows after adding an additional factor $n$.
  We do so by defining an injective function from a set $A$ to a set $B$.
  Set $A$ contains all strings $s_\cmb$. 
  Set $B$ contains all strings from $\{\ptdead,\ptalive,\ptfree\}^n$ which do not contain any subsequences of the form $(\ptdead,\ptfree,\dots,\ptfree,\ptdead)$, that is, one or more consecutive symbols $\ptfree$ enclosed by two symbols $\ptdead$.
  Such subsequences are called \emph{forbidden} henceforth.
  The bound on $|A|$ then follows from $|B| = \Theta(\alpha^n)$, an elementary counting problem whose proof we omit.

  Let us define such an injective function $f \colon A \rightarrow B$.
  For any $a \in A$ we construct $f(a)$ by the following process.
  We iterate over $a$ from left to right, and whenever we find a forbidden subsequence $(\ptdead,\ptfree,\dots,\ptfree,\ptdead)$ we replace it by a sequence $(\ptalive,\ptdead,\dots,\ptdead,\ptalive)$ of the same length.
  For example, if $a = (\ptdead,\ptfree,\ptdead,\ptdead,\ptdead,\ptfree,\ptdead)$ then $f(a) = (\ptalive,\ptdead,\ptalive,\ptdead,\ptalive,\ptdead,\ptalive)$.
  However, if $a = (\ptdead,\ptfree,\ptdead,\ptfree,\ptdead)$ then $f(a) = (\ptalive,\ptdead,\ptalive,\ptfree,\ptdead)$ because the second forbidden subsequence in $a$ is no longer a forbidden subsequence after the first one has been replaced.

  Clearly, $f$ is a function $A \rightarrow B$.
  It only remains to prove injectivity.
  Towards a contradiction, assume thus that $a \neq a'$ satisfy $f(a) = f(a') \eqqcolon b$.
  Let $i$ be the smallest index with $a_i \neq a_i'$.
  We distinguish the three cases $b_i = \ptfree$, $b_i = \ptalive$ and $b_i = \ptdead$.

  For $b_i = \ptfree$, observe that the function $f$ never uses the symbol $\ptfree$ to replace an entry in $a$ or $a'$.
  Hence, $a_i = a_i' = \ptfree$, a contradiction.

  For $b_i = \ptalive$, we may assume without loss of generality that $a_i = \ptdead$ and $a_i' = \ptalive$.
  Furthermore, $a_i$ must be either the first or last letter in a forbidden subsequence in $a$ that is replaced under $f$.
  From minimality of $i$ it follows that it is the first letter.
  Since the following argument generalizes to larger forbidden subsequences, we can now assume for reasons of simplicity that we have $(a_i,a_{i+1},a_{i+2}) = (\ptdead,\ptfree,\ptdead)$ and $(b_i,b_{i+1},b_{i+2}) = (\ptalive,\ptdead,\ptalive)$.
  There are two possibilities for the corresponding letters in $a'$.
  \begin{itemize}
    \item
      If $(a_i',a_{i+1}',a_{i+2}') = (\ptalive,\ptdead,\ptalive)$, then observe that the points $\pt_i$, $\pt_{i+1}$, $\pt_{i+2}$ from the set $\pts$ must form a left turn, that is, $\pt_{i+1}$ is strictly below the segment with endpoints $\pt_i$ and $\pt_{i+2}$.
      The reason is that $a' = s_{\cmb'}$ corresponds to a combination $\cmb'$ in $\cmbscf{\segs}$, and there must be a segment in $\cmb'$ which has $\pt_{i+1}$, but neither $\pt_{i}$ nor $\pt_{i+2}$, in its lower shadow.
      On the other hand, $(a_i,a_{i+1},a_{i+2}) = (\ptdead,\ptfree,\ptdead)$ implies that the points $\pt_i$, $\pt_{i+1}$, $\pt_{i+2}$ form a right turn, that is, $\pt_{i+1}$ is strictly above the segment with endpoints $\pt_i$ and $\pt_{i+2}$.
      Again, this holds because also $a = s_\cmb$ corresponds to a combination $\cmb$, and there must be a segment in $\cmb$ which has both $\pt_{i}$ and $\pt_{i+2}$, but not $\pt_{i+1}$, in its lower shadow.
      We have derived a contradiction because these two arrangements of $\pt_i$, $\pt_{i+1}$ and $\pt_{i+2}$ are mutually exclusive.
    \item
      If $(a_i',\dots,a_{i+3}') = (\ptalive,\ptdead,\ptdead,\ptfree)$ and $(b_i,\dots,b_{i+3}) = (\ptalive,\ptdead,\ptalive,\ptdead)$, that is, index $i+2$ is the beginning of a forbidden subsequence in $a'$, then we further obtain $(a_i,\dots,a_{i+3}) = (\ptdead,\ptfree,\ptdead,\ptdead)$.
      A contradiction can be derived in the same manner as in the previous case.
      The symbols in $a'$ imply that the points $\pt_i$, $\pt_{i+1}$ and $\pt_{i+3} $ form a left turn because there exists a segment which has $\pt_{i+1}$ and $\pt_{i+2}$, but neither $\pt_i$ nor $\pt_{i+3}$, in its lower shadow.
      The symbols in $a$ imply that $\pt_i$, $\pt_{i+1}$ and $\pt_{i+3}$ form a right turn because there exists a segment which has $\pt_i$, $\pt_{i+2}$ and $\pt_{i+3}$, but not $\pt_{i+1}$, in its lower shadow.
  \end{itemize}

  For $b_i = \ptdead$, we may assume that $a_i = \ptfree$, $a_i' = \ptdead$, and that in $a$ there exists a forbidden subsequence that starts before index $i$ and ends after index $i$ and that is replaced under $f$.
  Again for simplicity we assume $(a_{i-1},a_i,a_{i+1}) = (\ptdead,\ptfree,\ptdead)$ and $(b_{i-1},b_i,b_{i+1}) = (\ptalive,\ptdead,\ptalive)$, longer sequences can be treated similarly.
  By minimality of $i$ we get $a_{i-1}' = a_{i-1} = \ptdead$.
  The only possible way for $a_{i-1}' = \ptdead$ to be replaced by $b_{i-1} = \ptalive$ under $f$ is if $a_{i-1}'$ is the last letter in a forbidden subsequence in $a'$.
  Therefore, $(a_{i-3}',a_{i-2}',a_{i-1}') = (\ptdead,\ptfree,\ptdead)$ and $(b_{i-3},b_{i-2},b_{i-1}) = (\ptalive,\ptdead,\ptalive)$, again without loss of generality.
  By minimality of $i$ we further get $(a_{i-3},a_{i-2},a_{i-1}) = (\ptdead,\ptfree,\ptdead)$.
  In summary, we have derived $(a_{i-3},\dots,a_{i+1}) = (\ptdead,\ptfree,\ptdead,\ptfree,\ptdead)$ and $(b_{i-3},\dots,b_{i+1}) = (\ptalive,\ptdead,\ptalive,\ptdead,\ptalive)$.
  This is a contradiction because at most one of these two forbidden subsequences in $a$ is replaced under $f$.
\end{proof}

\subsection{Spanning Trees}

Recall the definition of $\agraphst{\pts}$ from \cref{subsec:st}.
The following lemma is all that is left to complete the proof of \cref{thm:st}.

\begin{lemma}\label{lem:stbound}
  For any $\pts$, there exists a subgraph of $\agraphst{\pts}$ induced by a vertex subset $V \subseteq \quotient{\cmbsst}{\parteq}$ satisfying $|V| = O(\alpha^nn)$, where $\alpha \lessapprox 7.04313$, and such that $V$ contains all vertices that appear on at least one $\cgsource$-$\cgsink$ path in $\agraphst{\pts}$.
  Moreover, given a vertex of $\agraphst{\pts}$, we can decide in time $O(n)$ whether it belongs to $V$.
\end{lemma}

Note that since we are able to recognize vertices from the set $V$ efficiently, the smaller subgraph of $\agraphst{\pts}$ can also be constructed efficiently bottom-up, simply by discarding any encountered vertices that do not belong to $V$.

\begin{proof}
  We begin by specifying the eight colors $\ptone,\pttwo,\dots,\pteight$ that were introduced only informally in \cref{subsec:st}.
  Given a combination $\cmb$ in $\cmbslst$, these colors are assigned to the points in $\pts$ as follows.
  If $\pt \in \clow{\cmb}$ then $\pt$ has color $\ptone$, that is, $\ptone = \ptdead$ in the original color scheme.
  Otherwise, and if additionally $\pt \not\in \upts{\cmb}$, then $\pt$ has color $\pteight$, that is, $\pteight = \ptfree$.
  As follows, all remaining points are assigned one of the colors $\pttwo,\dots,\ptseven$, which thus correspond to the original $\ptalive$.
  Color $\pttwo$ means that $\pt$ has out-degree 0 and exposes no drain.
  Color $\ptthree$ means that $\pt$ has out-degree 0 and exposes a drain to the left.
  Color $\ptfour$ means that $\pt$ has out-degree 0 and exposes a drain to the right.
  Color $\ptfive$ means that $\pt$ has out-degree 1 and exposes no drain.
  Color $\ptsix$ means that $\pt$ has out-degree 1 and exposes a drain to the left.
  Color $\ptseven$ means that $\pt$ has out-degree 1 and exposes a drain to the right.
  
  Similar to the proof of \cref{lem:cf:bound}, let $s_\cmb$ be the string from $\{\ptone,\dots,\pteight\}^n$ that corresponds to a given combination $\cmb$ in $\cmbslst$.
  The present lemma is a consequence of the following three observations, which all impose restrictions on certain patterns in $s_\cmb$.\\
  
  \hspace*{-\parindent}%
  \begin{minipage}{0.8\textwidth}
    \begin{itemize}
    \item For two consecutive points $\pt_i$ and $\pt_{i+1}$, it cannot be that both $\pt_i$ and $\pt_{i+1}$ have out-degree 1 and, at the same time, $\pt_i$ exposes a drain to the left and $\pt_{i+1}$ exposes a drain to the right. 
    As illustrated on the right, if that were the case then there would be either a crossing or a finite face of out-degree 0, both in contradiction with the definition of $\cmbslst$.
    Hence, in $s_\cmb$ we will never observe the pattern $(\ptsix,\ptseven)$, that is, color $\ptsix$ directly followed by color $\ptseven$.
    Additionally, if $\pt_i$ and $\pt_{i+k}$ (taking the role of $\pt_{i+1}$) are separated by any number of points of degree 0, the same argument still applies.
    Hence, we can further rule out the pattern $(\ptsix,\pteight,\dots,\pteight,\ptseven)$.
    \end{itemize}
  \end{minipage}
  \begin{minipage}{0.2\textwidth}
    \centering
    \input{fig/stpattern1}
  \end{minipage}
  \vspace{\belowdisplayskip}

  \hspace*{-\parindent}%
  \begin{minipage}{0.8\textwidth}
    \begin{itemize}
    \item Similarly, for two consecutive points $\pt_i$ and $\pt_{i+1}$, it cannot be that both $\pt_i$ and $\pt_{i+1}$ have out-degree 0 and, at the same time, $\pt_i$ exposes a drain to the left and $\pt_{i+1}$ exposes a drain to the right.
    Since both points have out-degree 0 they cannot be connected by an edge, and thus we necessarily get a crossing, as illustrated on the right.
    Hence, the pattern $(\ptthree,\ptfour)$ cannot occur in $s_\cmb$.
    Using the same argument we can further rule out the pattern $(\ptthree,\pttwo,\dots,\pttwo,\ptfour)$, since in such a configuration there will always be two consecutive points $\pt_i$ and $\pt_{i+1}$ of out-degree 0 such that $\pt_i$ is the left endpoint of an edge and $\pt_{i+1}$ is the right endpoint of another edge.
    Similar to the first observation, any additional points with degree 0 do not interfere with this argument.
    The same can be said for any points that are in the lower shadow of $\cmb$.
    Hence, we can rule out the pattern $(\ptthree,\ptone|\pttwo|\pteight,\dots,\ptone|\pttwo|\pteight,\ptfour)$, where ``$|$'' indicates an arbitrary choice between several options.
    \end{itemize}
  \end{minipage}
  \begin{minipage}{0.2\textwidth}
    \centering
    \input{fig/stpattern2}
  \end{minipage}
  \vspace{\belowdisplayskip}
  
  Lastly, we describe a pattern which might actually occur in $s_\cmb$.
  However, for any such occurrence we will prove that the corresponding vertex $\partclass{\cmb}$ is not contained in any $\cgsource$-$\cgsink$ path in $\agraphst{\pts}$.
  Therefore, the set $V$ from the lemma may safely be defined as the subset of all vertices of $\agraphst{\pts}$ that do not contain this final pattern.\\
  
  \hspace*{-\parindent}%
  \begin{minipage}{0.8\textwidth}
    \begin{itemize}
    \item Let $\pt_i$ and $\pt_{i+1}$ be two consecutive points such that $\pt_i$ exposes a drain to the right and $\pt_{i+1}$ exposes a drain to the left.
    Then, clearly, the out-degree of the infinite face in $\cmb$ is different from 0 and thus there is no direct connection from $\partclass{\cmb}$ to $\cgsink$ in $\agraphst{\pts}$.
    Moreover, any new segment $\unit$ that consumes the drain exposed by $\pt_i$ (that is, the border below $\pt_i$ becomes an out-border of the new face directly below $\unit$) must also consume the drain exposed by $\pt_{i+1}$, as illustrated on the right.
    However, the face below $\unit$ having out-degree at least 2 contradicts the definition of $\cmbslst$.
    Therefore, the two drains exposed by $\pt_i$ and $\pt_{i+1}$ can never be consumed and the sink $\cgsink$ can thus never be reached.
    Hence, we can safely discard any vertices with the pattern $(\ptfour|\ptseven,\ptthree|\ptsix)$.
    Additionally, suppose that $\pt_i$ and $\pt_{i+k}$ (taking the role of $\pt_{i+1}$) are separated by any number of points of degree 0 or points in the lower shadow of $\cmb$.
    Then, no matter how a new segment $\unit$ is added, we again get a pair of drains facing each other and all points in between are either of degree 0 or in the lower shadow.
    Hence, we can further discard any vertices with the pattern $(\ptfour|\ptseven,\ptone|\pteight,\dots,\ptone|\pteight,\ptthree|\ptsix)$.
    \end{itemize}
  \end{minipage}
  \begin{minipage}{0.2\textwidth}
    \centering
    \input{fig/stpattern3}
  \end{minipage}
  \vspace{\belowdisplayskip}
  
  Let $A$ be the subset of $\{\ptone,\dots,\pteight\}^n$ containing only strings without subsequences belonging to the  three patterns from the above observations.
  Then, it can be shown that $|A| = \Theta(\alpha^n)$ using standard techniques, and the lemma follows.
\end{proof}

\subsection{Spanning Cycles}

Finally, recall the definition of $\agraphsc{\pts}$ from \cref{subsec:sc}.
The following lemma concludes the proof of \cref{thm:sc}.

\begin{lemma}
  For any $\pts$, there exists a subgraph of $\agraphsc{\pts}$ induced by a vertex subset $V \subseteq \quotient{\cmbslsc}{\parteq}$ satisfying $|V| = O(\alpha^nn)$, where $\alpha \lessapprox 5.61804$, and such that $V$ contains all vertices that appear on at least one $\cgsource$-$\cgsink$ path in $\agraphsc{\pts}$.
  Moreover, given a vertex of $\agraphsc{\pts}$, we can decide in time $O(n)$ whether it belongs to $V$.
\end{lemma}

\begin{proof}
  We again start by giving specifications for the colors $\ptone,\pttwo,\dots,\ptsix$ that are assigned to the points in $\pts$ for any given combination $\cmb$ in $\cmbslsc$.
  Color $\ptone$ means that the corresponding point $\pt$ has degree 2 and exposes no drain.
  Color $\pttwo$ means that $\pt$ has degree 2 and exposes a drain.
  Color $\ptthree$ means that $\pt$ has degree 1 and exposes no drain.
  Color $\ptfour$ means that $\pt$ has degree 1 and exposes a drain to the left.
  Color $\ptfive$ means that $\pt$ has degree 1 and exposes a drain to the right.
  Color $\ptsix$ means that $\pt$ has degree 0.
  
  The following observation is similar to the third observation in the proof of \cref{lem:stbound}.
  Let $\cmb$ be a combination in $\cmbslsc$ and let $\pt_i$ and $\pt_{i+1}$ be two consecutive points.
  If $\pt_i$ exposes a drain to the right and $\pt_{i+1}$ exposes a drain to the left in $\cmb$, then $\cmb$ cannot be augmented in such a way that the infinite face has out-degree $0$ without creating finite faces with out-degree larger than $1$.
  In other words, the vertex $\partclass{\cmb}$ does not appear on any $\cgsource$-$\cgsink$ path in $\agraphsc{\pts}$ and may safely be discarded.
  The same holds for two not necessarily consecutive points $\pt_i$ and $\pt_{i+k}$ such that again $\pt_i$ exposes a drain to the right and $\pt_{i+k}$ exposes a drain to the left, and such that all points in between $\pt_i$ and $\pt_{i+k}$ have either degree 0 or 2.
  No matter how a new segment $\unit$ is added to $\cmb$, we will again end up with a combination with two exposed drains facing each other and only points of degree 0 or 2 in between.
  
  The subset $V$ of vertices of $\agraphsc{\pts}$ can thus be defined as follows.
  We simply exclude all vertices which contain the pattern $(\ptfive,\ptone|\pttwo|\ptsix,\dots,\ptone|\pttwo|\ptsix,\ptfour)$.
  Let now $A$ be the subset of $\{\ptone,\dots,\ptsix\}^n$ containing only strings without subsequences belonging to the above pattern.
  Then, $|A| = \Theta(\alpha^n)$, and the lemma follows.
\end{proof}

%% file: fig/stpattern1.tex
\begin{tikzpicture}[xscale=1.25]
  \begin{scope}[yshift=45]
    \node[alive] (p1) at (0,0) {};
    \node[alive] (p2) at (1,0) {};
    \draw[dotted,dloriented] (p1) -- ++(0,-0.6);
    \draw[dotted,droriented] (p2) -- ++(0,-0.6);
    \draw[oriented] (p1) -- ++(0.6,-0.6);
    \draw[oriented] (p2) -- ++(-0.6,-0.6);
  \end{scope}
  \begin{scope}[yshift=0]
    \fill[fill=black!20!white] (0,0) -- (1,0) -- (0,-0.6) -- cycle;
    \node[alive] (p1) at (0,0) {};
    \node[alive] (p2) at (1,0) {};
    \draw[dotted,dloriented] (p1) -- ++(0,-0.6);
    \draw[dotted,droriented] (p2) -- ++(0,-0.6);
    \draw[oriented] (p1) -- (p2);
    \draw[oriented] (p2) -- ++(-1.1,-0.66);
  \end{scope}
  \begin{scope}[yshift=-45]
    \fill[fill=black!20!white] (0,0) -- (1,0) -- (1,-0.6) -- cycle;
    \node[alive] (p1) at (0,0) {};
    \node[alive] (p2) at (1,0) {};
    \node[free] at (0.333,0.3) {};
    \node[free] at (0.6667,0.3) {};
    \draw[dotted,dloriented] (p1) -- ++(0,-0.6);
    \draw[dotted,droriented] (p2) -- ++(0,-0.6);
    \draw[oriented] (p1) -- ++(+1.1,-0.66);
    \draw[oriented] (p2) -- (p1);
  \end{scope}
  \begin{scope}[yshift=-75]
    \node {};
  \end{scope}
\end{tikzpicture}

%% file: fig/stpattern2.tex
\begin{tikzpicture}[xscale=1.25]
  \begin{scope}[yshift=45]
    \node[alive] (p1) at (0,0) {};
    \node[alive] (p2) at (1,0) {};
    \draw[dotted,dloriented] (p1) -- ++(0,-0.6);
    \draw[dotted,droriented] (p2) -- ++(0,-0.6);
    \draw[oriented] (0.6,-0.6) -- (p1);
    \draw[oriented] (0.4,-0.6) -- (p2);
  \end{scope}
  \begin{scope}[yshift=0]
    \node[alive] (p1) at (0,0) {};
    \node[alive] (p2) at (1,0) {};
    \node[alive] (p3) at (0.5,0) {};
    \draw[oriented] (0.3,-0.6) -- (p1);
    \draw[oriented] (0.7,-0.6) -- (p2);
    \draw[oriented] (0.2,-0.6) -- (p3);
    \draw[oriented] (0.8,-0.6) -- (p3);
    \draw[dotted,dloriented] (p1) -- ++(0,-0.6);
    \draw[dotted,droriented] (p2) -- ++(0,-0.6);
    \draw[dotted,droriented] (p3) -- ++(0,-0.6);
  \end{scope}
  \begin{scope}[yshift=-45]
    \node[alive] (p1) at (0,0) {};
    \node[alive] (p2) at (1,0) {};
    \node[alive] (p3) at (0.5,0) {};
    \draw[dotted,dloriented] (p1) -- ++(0,-0.6);
    \draw[dotted,droriented] (p2) -- ++(0,-0.6);
    \draw[dotted,droriented] (p3) -- ++(0,-0.6);
    \draw[oriented] (0.5,-0.6) -- (p1);
    \draw[oriented] (0.7,-0.6) -- (p2);
    \draw[oriented] (0.8,-0.6) -- (p3);
  \end{scope}
  \begin{scope}[yshift=-90]
    \node[alive] (p1) at (0,0) {};
    \node[alive] (p2) at (1,0) {};
    \node[free] at (0.4,0.2) {};
    \node[free] at (0.6,0.2) {};
    \draw[dotted,dloriented] (p1) -- ++(0,-0.6);
    \draw[dotted,droriented] (p2) -- ++(0,-0.6);
    \draw[oriented] (0.8,-0.6) node[dead] {} -- (p1);
    \draw[oriented] (0.2,-0.6) node[dead] {} -- (p2);
  \end{scope}
  \begin{scope}[yshift=-120]
    \node {};
  \end{scope}
\end{tikzpicture}

%% file: fig/stpattern3.tex
\begin{tikzpicture}[xscale=0.5]
  \begin{scope}[yshift=50]
    \node[alive] (p1) at (0,0) {};
    \node[alive] (p2) at (1,0) {};
    \draw[dotted,droriented] (p1) -- ++(0,-0.6);
    \draw[dotted,dloriented] (p2) -- ++(0,-0.6);
    \draw[oriented] (-1,0) -- (p1);
    \draw[oriented] (p2) -- ++(+1,0);
  \end{scope}
  \begin{scope}[yshift=0]
    \fill[fill=black!20!white] (0,0) -- (1,0) -- (1,-0.6) -- (0,-0.6) -- cycle;
    \node[alive] (p1) at (0,0) {};
    \node[alive] (p2) at (1,0) {};
    \draw[dotted,droriented] (p1) -- ++(0,-0.6);
    \draw[dotted,dloriented] (p2) -- ++(0,-0.6);
    \draw[oriented] (-1,0) -- (p1);
    \draw[oriented] (p2) -- ++(+1,0);
    \draw[oriented] (p1) -- node[above] {$\unit$} (p2);
  \end{scope}
  \begin{scope}[yshift=-50]
    \fill[fill=black!20!white] (0,0) -- (2,0.5) -- (2,0) -- (1,0) -- (1,-0.6) -- (0,-0.6) -- cycle;
    \node[alive] (p1) at (0,0) {};
    \node[dead] (p2) at (1,0) {};
    \draw[dotted,droriented] (p1) -- ++(0,-0.6);
    \draw[dotted,dloriented] (p2) -- ++(0,-0.6);
    \draw[oriented] (-1,0) -- (p1);
    \draw[oriented] (p2) -- ++(+1,0);
    \draw[oriented] (p1) -- node[above] {$\unit$} (2,0.5);
  \end{scope}
  \begin{scope}[yshift=-100]
    \node[alive] (p1) at (-0.3,0) {};
    \node[alive] (p2) at (1.3,0) {};
    \node[free] (p3) at (0.2333,0.25) {};
    \node[free] (p4) at (0.7666,0.25) {};
    \node[dead] (p5) at (0.2333,-0.7) {};
    \node[dead] (p6) at (0.7666,-0.7) {};
    \draw[dotted,droriented] (p1) -- ++(0,-0.5);
    \draw[dotted,dloriented] (p2) -- ++(0,-0.5);
    \draw[oriented] (-1,0) -- (p1);
    \draw[oriented] (p2) -- (2,0);
    \draw[oriented] (-1,-0.5) -- (2,-0.5);
  \end{scope}
  \begin{scope}[yshift=-150]
    \node[alive] (p1) at (-0.3,0) {};
    \node[dead] (p2) at (1.3,0) {};
    \node[free] (p3) at (0.2333,0.25) {};
    \node[alive] (p4) at (0.7666,0.25) {};
    \node[dead] (p5) at (0.2333,-0.7) {};
    \node[dead] (p6) at (0.7666,-0.7) {};
    \draw[dotted,droriented] (p1) -- ++(0,-0.5);
    \draw[dotted,dloriented] (p2) -- ++(0,-0.5);
    \draw[dotted,dloriented] (p4) -- ++(0,-0.75);
    \draw[oriented] (-1,0) -- (p1);
    \draw[oriented] (p2) -- (2,0);
    \draw[oriented] (p4) -- node[above] {$\unit$} (2,0.25);
    \draw[oriented] (-1,-0.5) -- (2,-0.5);
  \end{scope}
  \begin{scope}[yshift=-175]
    \node {};
  \end{scope}
\end{tikzpicture}

%% file: acknowledgement.tex
\section{Acknowledgements}

Most of the results presented in this paper were originally obtained in the author's master's thesis \cite{W13} in a, however, substantially less concise format.
The author is deeply grateful and indepted to his then supervisor and now PhD advisor Emo Welzl.
Special thanks go to Raimund Seidel for presenting his beautiful algorithm for counting triangulations at ETH Z\"urich in fall 2012, which undoubtedly marked the beginning of the success of that thesis.